\documentclass{article}
\usepackage[left=1.25in,top=1.25in,right=1.25in,bottom=1.25in,head=1.25in]{geometry}
\usepackage{amsmath, amsthm, amsfonts, amssymb, amsbsy,  graphicx,
  enumerate}
\usepackage{enumitem}
\usepackage{verbatim,float,url}
\usepackage{graphicx,epsfig,psfrag}
\usepackage{color}
\usepackage[authoryear]{natbib}

\usepackage{pdfpages}
\usepackage{bbm}
\usepackage{url}

\usepackage{bm} 
\usepackage{algorithm}
\usepackage{algpseudocode}

\def\E {\rm E}

\newcommand{\tblue}[1]{\textcolor{blue}{{#1}}}

\providecommand{\minimize}{\mathop{\rm minimize}}

\renewcommand{\P}{\mathbb{P}}

\newcommand{\mc}{\mathcal{C}}

\newcommand{\tL}{\tilde{L}}

\newtheorem{thm}{Theorem}
\newtheorem{lmm}{Lemma}
\newtheorem{cor}{Corollary}

\theoremstyle{definition}

\newcommand{\cov}{\mathrm{Cov}}

\newcommand{\md}{\mathcal{D}}

\newcommand{\R}{\mathbb{R}}

\newcommand{\var}{\mathrm{Var}}

\newcommand{\sign}{\operatorname{sign}}

\newcommand{\tx}{\tilde{X}}

\newcommand{\aols}{\hat{\beta}_0^{\mathrm{OLS}}}
\newcommand{\aur}{\hat{\beta}_0^{\mathrm{UR}}}
\newcommand{\bols}{\hat{\beta}^{\mathrm{OLS}}}
\newcommand{\bur}{\hat{\beta}^{\mathrm{UR}}}
\newcommand{\buni}{\hat{\beta}^{\mathrm{uni}}}

\newcommand{\bbuni}{\beta^{\mathrm{uni}}}
\newcommand{\muols}{\hat{\mu}^{\mathrm{OLS}}}
\newcommand{\muur}{\hat{\mu}^{\mathrm{UR}}}
\newcommand{\gols}{\hat{\gamma}^{\mathrm{OLS}}}
\newcommand{\gur}{\hat{\gamma}^{\mathrm{UR}}}
\newcommand{\hl}{\hat{\lambda}}

\begin{document}

\title{Univariate-Guided  Sparse Regression }
\author{Sourav Chatterjee$^{1,3} $  \and  Trevor Hastie$^{2,3}$ \and 
 Robert Tibshirani$^{2,3}$}
 
\date
{\normalsize
$^1$Department of Mathematics \\
$^2$Department of Biomedical Data Science \\
$^3$Department of Statistics \\
Stanford University\\
}

\maketitle

\begin{abstract}
In this paper, we introduce ``uniLasso''--- a novel statistical method for regression.  
This  two-stage approach preserves the signs of the univariate coefficients and leverages their magnitude. Both of these properties are attractive for stability and interpretation of the model.
Through comprehensive simulations and applications to real-world datasets, we demonstrate that uniLasso outperforms lasso in various settings, particularly in terms of sparsity and model interpretability. We prove asymptotic support recovery and mean-squared error
consistency under a set of conditions different from the well-known irrepresentability conditions for the lasso.
Extensions to generalized linear models (GLMs) and Cox regression are also discussed.

\end{abstract}

\section{Introduction}
\label{sec:intro}
High-dimensional regression and classification problems are ubiquitous across fields such as genomics, finance, and social sciences. In these settings, the lasso (Least Absolute Shrinkage and Selection Operator) has emerged as a widely-used methodology due to its ability to simultaneously perform variable selection and regularization, yielding sparse and interpretable models. Despite its widespread use, lasso has certain limitations, including sensitivity to correlated predictors and inclusion of spurious features in the final model. There have been many related proposals
including MC+ \citep{zhang2010}, the elastic net  \citep{ZH2005},  the adaptive lasso \citep{zou2006a}, and SparseNet \citep{MFH2010}.

To address these challenges, we propose uniLasso, a novel regression methodology.  In this work, we aim to achieve two primary goals: predictive accuracy and (especially)  interpretability. 
To achieve this,  uniLasso integrates marginal  (univariate) information into a coherent multivariate framework: our method preserves the signs of the univariate coefficients and leverages their magnitude. This approach not only enhances predictive accuracy but also provides insights into the relative importance of predictors.

Here is an example, taken from unpublished work with collaborators of the third author.
(Since the work is unpublished, we do not give  background details). The dataset is from cancer proteomics, with 559 biomarkers measured on 81 patients, 20 healthy and 61 with cancer. We divided the data into 70\% training and 30\% test, and applied
both the lasso and uniLasso to the training set.
Both methods performed well, giving cross-validation error of about 5\%, and 4 and 3 errors respectively out of 30 test samples.
Table~\ref{tab:prot} shows the results of lasso and uniLasso applied to these data, using cross-validation to choose the model. In the top table there are two sign changes from
univariate to multivariate: \#12 goes from mildly protective to a negative indicator, while  \#11 goes in the opposite direction.

By design, uniLasso in the bottom table
has no sign changes. Further, its univariate coefficients are much larger (in absolute value) than those of the lasso. This makes it a more credible and stable model.
%We investigate stability in Section \ref{subsec:adverse}.

The remainder of this paper is organized as follows. Section \ref{sec:uni} provides a detailed description of the uniLasso methodology. We examine the relationship
between uniLasso and the adaptive lasso in Section \ref{sec:adapt}, and in Section \ref{sec:sparsity} we study the underlying reason for the strong sparsity produced by uniLasso. We compare uniLasso to lasso on a car price data set in Section \ref{sec:car}: here we show the utility
of unregularized uniLasso as an alternative to least squares.   In
Section \ref{sec:ortho} we look at the orthonormal design case, where
we can derive an explicit expression for the uniLasso
solution. Section \ref{sec:theory} presents some theoretical results
on support recovery and mean square error consistency. Section
\ref{sec:sim} presents simulation studies comparing uniLasso to lasso
and other benchmark methods, while Section \ref{sec:unireg} examines
uniReg, the unregularized ($\lambda=0$) case. Section \ref{sec:real}
applies uniLasso to real-world datasets, and in
Section~\ref{sec:multiclass} we adapt uniLasso to multiclass
classification problems.
% Section \ref{sec:moreexamples}  we present an ``adversarial attack''
% view, and multiclass and GWAS examples.
Section \ref{sec:external} explores the setting where additional data is available: not the complete data,  but just the univariate scores.  In Section \ref{sec:cv}
we study whether cross-validation works properly in the uniLasso setting. We discuss uniLasso for GLMs and the Cox Survival Model, as well as computation in 
Sections \ref{sec:UniLassocox} and \ref{sec:effic-comp-unin}. % NOTE
%Section 11.
Lastly,  Section \ref{sec:discussion} %NOTE
%Section 12
 concludes with a discussion of future directions.

\begin{table}[h]
\centering
    \begin{tabular}{cccc}
    \multicolumn{3}{c}{\em Model chosen by lasso }\\
 
Biomarker & Univariate LS Coefficient & Lasso Coefficient \\
   \hline\\
  1 & -49.24 & -18.92 \\ 
  2 & -37.65 & -5.42 \\ 
  3 & 25.27 & 37.49 \\ 
  4 & -24.74 & -11.26 \\ 
  5 & 22.91 & 14 \\ 
  6 & -22.38 & -12.75 \\ 
  7 & -17.97 & -11.91 \\ 
  8 & 13.02 & 0.79 \\ 
  9 & -12.71 & -1.85 \\ 
  10 & -10.14 & -2.32 \\ 
  11 & \tblue{1.81} & \tblue{-6.27}\\ 
  12 & \tblue{-0.34 }& \tblue{6.25} \\ 
\end{tabular}

 \begin{tabular}{ccc}
 \hline
 \hline\\
    \multicolumn{3}{c}{\em Model chosen by uniLasso }\\
 
Biomarker & Univariate LS Coefficient & uniLasso Coefficient\\
   \hline\\
  A & 93.13 & 13.82 \\ 
  B & -84.26 & -8.17 \\ 
  C & 77.69 & 15.86 \\ 
  D & 73.84 & 27.66 \\ 
  E & 69.15 & 16.21 \\ 
  F & 58.92 & 1.38 \\ 
  G & -53.46 & -10.5 \\ 
  H & -49.24 & -20.76 \\ 
  I & 25.27 & 10.48 \\ 
\hline
  \hline

 \hline
\end{tabular}
\vskip .2in
\caption{\em Proteomics study. Top table: Univariate least squares coefficients and lasso coefficients, for model chosen by lasso. The two sign changes from univariate to multivariate are shown in blue. Bottom table: same for model chosen by uniLasso.
Biomarkers are ordered by decreasing absolute value.
Most biomarkers are different in the two tables, with the exception of 1 and 3 in the first table corresponding to
H and I in the second table.}
\label{tab:prot}
\end{table}

\newpage
\section{Univariate-guided lasso}
\label{sec:uni}
\subsection{Our proposal}
\label{sec:our-proposal}
We assume the standard supervised learning setup:
we have  training features  and target $X_{n\times p},\; y_{n\times 1}$.
%, and test features $X_{test}$ of size $n_{test} \times  p$.
For now we assume that $y$  is quantitative and we fit a linear model using squared error loss;
later we discuss the binomial and other GLM families, as well as the
Cox model.

Our procedure has three simple steps, which we motivate here.
For interpretability and prediction accuracy,  we preprocess the
features in Step~1,  multiplying them by a robust version of their
univariate least-squares coefficient estimates. In Step~2 we fit a
non-negative lasso tuned by cross-validation, and Step~3 combines the
components of Steps~1 and 2 to produce a final model. Together these ensure that:
\begin{enumerate}
\item[(a)]the signs of the final coefficients agree with the
    signs of the univariate coefficients (or they are zero);
\item[(b)] features with larger univariate coefficients will tend to have larger coefficients in the final model.
\end{enumerate} 

\begin{algorithm}[H]
  \smallskip  
\centerline{\bf UniLasso algorithm}

\medskip
\begin{enumerate}
\item  For $j=1,2,\ldots,p$ compute the univariate intercepts and slopes   $(\hat\beta_{0j}, \hat\beta_j)$  and their leave-one-out (LOO) counterparts
 $(\hat\beta_{0j}^{-i}, \hat\beta_j^{-i}),\;i=1,\ldots,n.$
\item  {Fit the lasso} --- with an intercept, no standardization, and
  {non-negativity constraints  } --- to target $y$ and the
  univariate LOO fits as features

 \begin{equation*}
\minimize_{\theta} \;\left\{\frac{1}{n} \sum_{i=1}^n
  \Bigl(y_i-\theta_0-\sum_{j=1}^p {(\hat\beta_{0j}^{-i}
    +\hat\beta_j^{-i}x_{ij})} \theta_j \Bigr)^2 +\lambda
  \sum_{j=1}^p\theta_j\right\}\quad\mbox{ with  $\theta_j\geq 0
  \;\forall j $}.
% \label{eqn:nnlasso}
\end{equation*}
Select $\lambda$ by cross-validation.

%  \begin{equation*}
 %   \hat{\eta}(x) = \hat\theta_0 + \sum_{j=1}^p\hat\theta_j\hat\eta_j(x_j).
 % \end{equation*}
  
  \bigskip
%\rule{\textwidth}{1pt}. 

\item The final model can be written as  
%\begin{equation*}
$ \hat\eta(x)=\hat\gamma_0 + \sum_{j=1}^p\hat\gamma_j x_j$,
%\end{equation*}
with $\hat\gamma_j=\hat\beta_j\hat\theta_j$, and $\hat\gamma_0=\hat\theta_0+\sum_{\ell=1}^p\hat\beta_{0\ell}\hat\theta_\ell$.

\end{enumerate}

\end{algorithm}

  This procedure is computationally convenient: in Step~1 we can use
  efficient LOO formulas, and in Step~2 we can apply any efficient
  $\ell_1$ solver. Here we used the R language package {\tt glmnet}. Specifically, we use the 
  function {\tt cv.glmnet} to estimate the
  lasso path parameter, and have all of the functionality of {\tt
    glmnet} at our disposal. We provide a function {\tt cv.UniLasso}
  in the R package {\tt uniLasso}
  that implements this approach.

Note that  we {\em do not} standardize the features before applying
 the non-negative lasso in Step 2; the univariate LOO fits are all on the scale of the response.  From our knowledge  of  multiple
 linear regression, the first constraint --- agreement between
 univariate and multivariate signs --- may seem like an unreasonable restriction. However our belief is that in high-dimensions
  there are  likely to be a multitude of different models  that have about the same MSE as the ``optimal'' one chosen by the lasso. Hence it can make sense to choose one that is interpretable and sparser than that of lasso.

  These properties mean, for example, if the feature ``age'' has a
  positive univariate coefficient --- indicating increasing risk of the outcome variable (such as Alzheimer's disease), it will have a positive (or zero) coefficient in the final lasso model.
  And if age is strongly significant on its own,  it is more likely to be chosen in the multivariate model.

These conjectures are borne out by our numerical studies.
In simulations with varying problems sizes and SNR, and a number of
real datasets (see Section~\ref{sec:sim}), in  almost every case
uniLasso did no worse than the lasso in terms of out-of-sample MSE, and 
  delivered a substantially  sparser model. 
 %These real datasets were not ``cherry picked''   but were the first group of datasets that we tested.
We note that \cite{mein2012} studies sign-constrained least squares estimation for high-dimensional regression, which relates to our non-negativity constraint in our second step.

\begin{enumerate}[label={\bf Remark~\Alph*}., leftmargin=0pt, itemindent=!, labelwidth=*, align=left]
%\item In equation (\ref{eqn:nnlasso}), the first term has a multiplier of $1/n$. In other descriptions of the lasso,  the factor may be  set to  $1/2$; the {\tt glmnet package} uses $1/2n$.
% The different choices lead to equivalent problems and just differ in the scaling of the path parameter $\lambda$.
% In the theory of Section~\ref{sec:theory} we assume that the $1/n$ multiplier is used.

  \item
  The uniLasso procedure applies seamlessly to other GLMs
as well the Cox model. Indeed, all of the  models covered by the {\tt glmnet} package are at our disposal.
All we need are the separate fitted linear models and their LOO fit vectors in step 1, and then {\tt glmnet} can be directly applied.
The only challenge is find an (approximate)~LOO formula for the GLM
families. We give details of these computations in  Section~\ref{sec:effic-comp-unin}.
\item
  UniLasso can be thought of  as a version of stacked regression, or \emph{Stacking} \citep{Wo92,Br97a}.
Stacking  is a two-stage procedure for combining the predictions from  a set of models (``learners''),
in order to get improved predictions.
It works as follows.
A set of base models is trained on the training data; these models can be of the same type (e.g. gradient boosting)
 or different types (e.g., linear regression, decision trees, etc.).
Each model generates predictions, capturing specific aspects of the relationship between the predictors and the target variable.
A meta-model is then trained to combine the predictions of the base
models into a final prediction. This meta-model learns how to weigh
and integrate the LOO outputs of the base models to minimize the overall prediction error.  UniLasso is a special case of stacking where the individual learners are simple univariate regressions.
\item
  Why do we use the LOO univariate estimates $\hat\eta_j^{-i}=\hat\beta_{0j}^{-i} +\hat\beta_j^{-i}x_{ij}$ as
features in step~2, instead of  the univariate estimates
$\hat\eta_j^i=\hat\beta_{0j}+\hat\beta_jx_{ij}$?\footnote{In this case
  we could simply scale the features by
  their univariate slope coefficients and ignore the intercepts, which would
  be handled by the overall intercept in the model.} 
In traditional stacking  this is essential because the individual learners  can be of very different complexity.
In uniLasso, one would think that the learners (univariate regression) all have the same complexity so that the LOO versions are not needed.
Indeed, the theory in Section~\ref{sec:theory} is based on the LOO estimates but holds equally well for the usual (non-LOO) estimates.
However in practice we have found  that the LOO estimates lead to  greater sparsity and better performance, and hence we use them.
With the use of the usual (non-LOO)   univariate  estimates in uniLasso, the resulting estimator is closely related to the adaptive lasso and we explore this connection in Section \ref{sec:adapt}.
\item
  The non-negative garotte \citep{Br95} is another closely related method. It minimizes $\sum_i(y_i- \sum_{j}  c_j \hat\beta_j x_{ij})^2$ subject to $c_j \geq 0$ for all $j$,  and $\sum c_j \leq s$,
where $\hat\beta_j$ are the usual (multivariable) LS estimates. This is different from our proposal in an important way: our use of univariate LOO estimates in the first step.
The univariate coefficients lead to materially different solutions and also allow application to the $p>n$ scenario.
\end{enumerate}
\subsection{UniLasso with no-regularization}

We get an interesting special case of uniLasso if we set $\lambda= 0$, so that there is no $\ell_1$ regularization in Step~2.
That is, we solve the following::
 \begin{equation}
{\rm minimize}_{\theta} \;\left\{\frac{1}{n} \sum_{i=1}^n \Bigl(y_i-\theta_0-\sum_{j=1}^p \theta_j  {(\hat\beta_{0j}^{-i} +\hat\beta_j^{-i}x_{ij})}\Bigr)^2\right\}\quad\mbox{ with  $\theta_j\geq 0 \;\forall j $}.
\label{eqn:unireg}
 \end{equation}

 We call this ``uniReg'' for univariate-guided regression and it represents an interesting alternative to the usual least squares estimates.
For example, the non-negative constraint can still provide sparsity even though there is no $\ell_1$ penalty. If $p>n$ there can be multiple solutions; in this case we look for the sparsest solution by computing the limiting uniLasso solution as $\lambda\downarrow 0$.
%And it gives a unique solution even if $p>n$.
The standard error and distribution of the estimated coefficients can be estimated via the bootstrap.
We study uniReg  in detail in Section~\ref{sec:unireg}.

\subsection{Two examples --- one good, one bad}
\label{sec:twoexamples}
{\bf Homecourt. } We generated data with 100 observations, 30 standard
normal features with AR(1) correlation 0.8, 20\% sparsity and
non-negative coefficients in two stages, as follows:
\begin{eqnarray}
y'&\leftarrow&\sum_jx_j\beta_j+\sigma' z'
\end{eqnarray}
From these data we compute the $p=30$ univariate least squares coefficients
$\hat\beta_j^{uni}$ separately for each~$j$, and then generate $y$ as 
\begin{eqnarray}
y&\leftarrow&\sum_jx_j \hat\beta_j^{uni}\beta_j+\sigma z.
\label{eqn:homecourt}
\end{eqnarray}
The variance terms $\sigma'$ and $\sigma$ were chosen so that at both stages the  SNR
was 1. The idea here is that (\ref{eqn:homecourt}) roughly mimics the  model fit by uniLasso in its second stage.
\begin{figure}[hbtp]
\begin{center}
\includegraphics[width=5in]{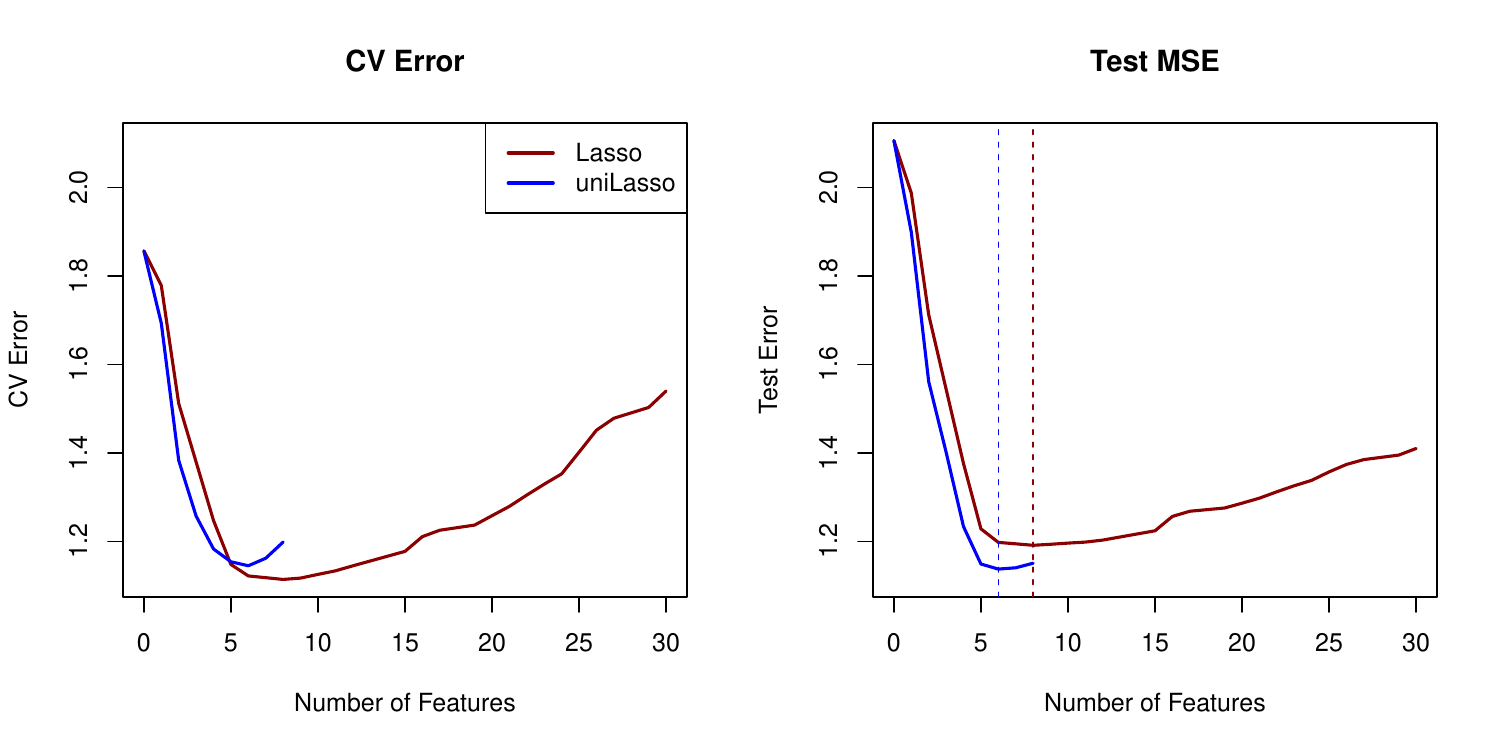}
\end{center}
\caption{\em Results  for homecourt example: CV and test set
  prediction error. The dashed vertical lines for each method correspond to the models
 chosen by CV.}
\label{fig:fig1}
\end{figure}

%Figure 1;
%the two quantities differ in expectation by $\sigma^2$
%Notice that UniLasso has much higher TPR and much lower FPR than lasso.

\begin{table}[hbtp]
\begin{center}
\begin{tabular}{rrrrr}
  \hline
 & MSE-lasso & MSE-uniLasso & Support-lasso & Support-uniLasso \\ 
  \hline
Mean & 1.098 & 1.077 & 7.660 & 4.790 \\ 
  se & 0.005 & 0.005 & 0.309 & 0.142 \\ 
   \hline
  \hline
 & TPR-lasso & TPR-uniLasso & FPR-lasso & FPR-uniLasso \\ 
  \hline
Mean & 0.737 & 0.700 & 0.135 & 0.025 \\ 
  se & 0.014 & 0.016 & 0.012 & 0.004 \\ 
   \hline
% from  sc300.r
% xtable(cbind(outerrmse[,1:2],out2[,1:2]))
%xtable(out3[,1:4])
%  & mse-lasso & mse-uniLasso & support-lasso & support-uniLasso \\ 
%   \hline
% mean & 19999.07 & 14000.44 & 10.00 & 6.00 \\ 
%   se & 1546.34 & 1096.42 & 0.42 & 0.31 \\ 
%    \hline
%   \hline
%  & TPR.lasso & FPR.lasso & TPR.uni & FPR.uni \\ 
%   \hline
% mean & 0.42 & 0.58 & 0.65 & 0.35 \\ 
%   se & 0.02 & 0.02 & 0.02 & 0.02 \\ 
%    \hline
   \end{tabular}
   \end{center}
   \caption{\em Test MSE, support, TPR (True Positive Rate) and FPR (False Positive Rate) for 100 simulations from  the setting of Figure~\ref{fig:fig1}.}
   \label{tab:tab1}
\end{table}
Figure~\ref{fig:fig1}
shows the CV and test error curves for the lasso and uniLasso.
% As expected, uniLasso has lower MSE, smaller support, and better true positive
%  and false positive rates.
We see that {\tt  uniLasso}  has test error a little below that of
{\tt lasso}, with a smaller active set.
%Table \ref{tab:tab1}. NOTE
Table~\ref{tab:tab1}
shows the result of 100 simulations from this setup.
We see that the same pattern emerges, and uniLasso exhibits a slightly
lower true positive rate and a much lower false positive rate than the lasso.

{\bf Counter-example.}  Here we take $n=100$, $p=20$, $x_1 \sim
N(0,1)$, $x_2=x_1+N(0,1)$, $\beta=(1,-.5, 0,0,\ldots 0)$, and $\mbox{error SD} =0.5$.
The remaining 18 features are standard normal.
This is an ``Achilles heel'' for uniLasso, as the (negative) sign of the population  coefficient for $x_2$ differs from its (positive) univariate sign.
As a result, the test MSE for uniLasso is about twice  that of lasso (detailed results are given in Section \ref{sec:sim}).
Clearly uniLasso fails badly here, but its high  CV error will alert
the user to this.  Motivated by this kind of example, we discuss a post-processing (``polish'') for
uniLasso in Section~\ref{sec:polish} that remedies this problem.

% \iffalse
% \begin{figure}[hbtp]
% \begin{center}
% \includegraphics[width=5in]{ex2_TH}

% \includegraphics[width=5in]{ex2a_TH}

% \end{center}
% \label{fig:fig2}

% \caption{\em Results for  simulated high SNR example, where the true
% number of nonzero coefficients is 100 and $p=1000$. Top panel: CV and test set prediction error. The dotted vertical lines shows the values of 
% $\lambda$,
% chosen by CV for each model.  Bottom panels:  true positives (green)
% and false positives (red).
% %Horizontal dotted lines: mean value of true nonzero coefficients for the chosen support set.
% }
% \end{figure}
% \fi
% %\newpage

\section{Relationship of uniLasso to the adaptive lasso}
\label{sec:adapt}

The adaptive lasso \citep{zou2006a} is defined by

\begin{equation}
\hat\gamma^*={\rm argmin} =\frac{1}{2} \sum_{i=1}^n (y_i-\gamma_0-\sum_{j=1}^p x_{ij}\gamma_j )^2 +\lambda \sum_{j=1}^p w_j|\gamma_j|
\end{equation}
where 
$w_j =1/|\hat \gamma_j|^\nu$. Here the vector $\hat\gamma$ is any root-$n$ consistent estimate of $\gamma$,
for example the least squares estimates. 

The original paper assumed $n>p$; with $p>n$, one cannot use the least-squares estimates.
One possibility is to use an initial estimate (from say ridge or lasso), solve the adaptive lasso,
and iterate. Another option is to use the univariate least-squares estimates; \citep{huang2008} consider a combination of these to achieve some theoretical guarantees. 

Let $\{\hat\beta_j\}_1^p$ be the univariate least squares coefficients. Consider the adaptive lasso with weights $w_j=1/|\hat\beta_j|,\;j=1,\ldots,p$. Then it is easy to show that this procedure is equivalent to
uniLasso with $\hat\eta_j^i= \hat\beta_{0j}+\hat\beta_jx_{ij}$ replacing their LOO versions
$\hat\eta_j^{-i}=\hat\beta_{0j}^{-i} +\hat\beta_j^{-i}x_{ij}$ in Step~2, and removing the non-negativity constraint.

 This procedure does not share some of the properties of
uniLasso. In particular, the final coefficients may not have the same signs as the univariate coefficients.
Importantly, in our simulations of Section~\ref{sec:sim}, it tends to
yield less sparse models and sometimes much higher MSE.

Alternatively, one could include the non-negativity constraints in this version of the adaptive lasso.
The following result shows that we obtain a procedure equivalent to  uniLasso, except for the use of LOO estimates in step 1.

{\noindent\bf Proposition 1.} Let $(\hat\beta_{0j},\hat\beta_j)$ be the
univariate least squares coefficients for variable $j$ in a $p$
variable linear model with data $\{(x_i,y_i)\}_1^n$. Let $\eta_j^i =\hat\beta_{0j}+\hat\beta_jx_{ij} $ be
the univariate linear fit for variable $j$ and observation $i$.

Then the following two problems are equivalent:

\begin{equation}
\label{eq:1}
\min_{\theta} \sum_{i=1}^n (y_i - \theta_0
-\sum_{j=1}^p\eta_j^i\theta_j)^2 + \lambda\sum_{j=1}^p |\theta_j|\quad
\mbox{ s.t. $\theta_j \geq 0\;\forall j$} 
\end{equation}

\begin{equation}
  \label{eq:2}
\min_{\gamma} \sum_{i=1}^n (y_i - \gamma_0
-\sum_{j=1}^px_{ij}\gamma_j)^2 + \lambda\sum_{j=1}^p\frac{ |\gamma_j|}{|\hat\beta_j|}\quad
\mbox{ s.t. $\sign(\gamma_j)=\sign(\hat\beta_j)\;\forall j$}
\end{equation}

The exact equivalence is obtained with
$\hat\gamma_j=\hat\theta_j\hat\beta_j$, and
$\hat\gamma_0=\hat\theta_0+\sum_{\ell=1}^p\hat\theta_\ell\hat\beta_{0\ell}$.

The proof is in the Appendix.
However our experiments with this version of adaptive lasso produced
models that were not nearly as sparse as the ones using LOO. In fact, in
Table~\ref{tab:ingredients} we see that enforcing sign constraints
led to less sparse models than  not enforcing sign constraints! 

% \iffalse

% Here is an example.
% Figure \ref{fig:adapt} shows the results from one realization of the ``low-SNR" setting of Section \ref{sec:sim}.
%  We see that 
% both versions of the adaptive lasso overfit, and cross-validation fails to detect this. 
% On the other hand, uniLasso chooses a null model and achieves a much lower MSE.
% The use of  the usual (non-LOO) univariate coefficients in the adaptive lasso causes the problem: it allows the procedure to overfit and cross-validation fails.
% The reason is that for computational efficiency, cross-validation is not wrapped around both steps (estimation of univariate coefficients and the lasso fit) but is only applied to the lasso fit.

% By using LOO fits, uniLasso avoids this problem, because step 1 is internally validated.
% We note that this poor performance of adaptive lasso does not always occur in this setting, but occurs often enough to hurt  its average performance, especially in support size.

% \begin{figure}[hbtp]
% \begin{center}
% \includegraphics[width=7in]{ex-adapt}
% \hskip 1in

% \end{center}
% \label{fig:adapt}
% \caption{\em CV curves from a low-SNR simulation with $n=300, p=1000$.
% The test MSE for each chosen model is given above each panel. The numbers across the top of each panel indicate the model size at each value of $\lambda$.
% \end{figure}
% \fi

Finally we note that \cite{candes2008} propose an iterated version of the adaptive lasso, in which the current solutions are used to define weights for the next iteration. They show that the method can enhance sparsity,
especially for signal recovery.

Next we examine the relative performance of uniLasso when we modify
the three main aspects of its design:
\begin{enumerate}[label=(\alph*)]
\item the use of LOO features $\hat\eta_j^{-i}$ versus non-LOO
  features $\hat\eta_j^i$,
\item the sign constraints, and
\item the scaling in Step~2 by the magnitude of the univariate
  coefficients, rather than just by their signs.
\end{enumerate}  
Note that the adaptive lasso is equivalent to uniLasso using non-LOO
estimates and no sign constraints.

Table~\ref{tab:ingredients} shows the results of 50 simulation runs with $n=300,\; p=1000$ and SNR=1, our ``medium-SNR''
setting detailed in Section~\ref{sec:sim}.

\begin{table}[ht]
\centering

\begin{tabular}{rrrrr}
\multicolumn{5}{c} {\bf LOO} \\
  \hline
 & lasso & uniLasso & uniLasso-noSign & uniLasso-noMag \\ 
  \hline
MSE & 0.55 & 0.59 & 0.61 & 0.60 \\ 
  se & 0.03 & 0.03 & 0.04 & 0.04 \\ 
   \hline
  \hline
Support & 53.78 & 15.32 & 66.38 & 40.12 \\ 
  se & 4.41 & 1.19 & 12.92 & 6.47 \\ 
   \hline
\end{tabular}
\centering
\begin{tabular}{rrrrr}
\multicolumn{5}{c} {\bf No LOO} \\
  \hline
 & lasso & uniLasso & uniLasso-noSign & uniLasso-noMag \\ 
  \hline
MSE & 0.55 & 0.61 & 0.57 & 0.61 \\ 
  se & 0.03 & 0.04 & 0.03 & 0.05 \\ 
   \hline
  \hline
Support & 53.78 & 36.84 & 24.48 & 48.98 \\ 
  se & 4.41 & 6.34 & 2.71 & 9.66 \\ 
   \hline
\end{tabular}
\caption{\em Results for 50 simulations from the  medium-SNR
  setting. The last two columns use variations of uniLasso: {\em
    noSign} removes the non-negative constraint, while {\em noMag}
  uses just the sign of the univariate estimates in Step 2 of uniLasso (not the magnitude). The top table uses LOO univariate coefficients (as in uniLasso) while the bottom table uses the usual (non-LOO) versions.}
\label{tab:ingredients}
\end{table}

We see that there is not much difference in test error, but uniLasso
produces the sparsest models by a large  margin. We also see the
curious result that with no LOO, the uniLasso {\em without} sign
constraints selects sparser models than {\em with} sign constraints.
\section{Why does the uniLasso produce such sparse solutions? }
\label{sec:sparsity}
As we saw in the previous section, the combination of LOO univariate estimates in Step (1)  and nonnegative constraints in step(2) seem to be the key to the strong
sparsity delivered by uniLasso. What is the explanation for this?
%\footnote{We thank Chris Habron for his contribution to this analysis.}
\footnote{We thank Chris Habron for  this analysis.}

The correlation between the responses $y_i$ and the univariate fitted
values $\hat{\beta}_{0j}+\hat\beta_jx_{ij}$ is always positive, and is the absolute value of the correlation between $y$ and the $j$th feature. 
When this correlation is large, then the correlation between $y_i$ and
the LOO fitted values $\hat\beta_{0j}^{-i} +\hat\beta_j^{-i}x_{ij}$ will also tend to
be positive.  But if the first correlation is positive  and small then the second correlation  is often negative.  
This could be explained by the fact that leaving an observation out ``pushes'' the fitted regression line away from that point, and if there wasn't much correlation to start with that will be enough to tip the correlation of the LOO feature to be negative.\footnote{In the extreme case of zero correlation between $y$ and the $j$th feature, the univariate fit is $\hat\beta_{0j}=\bar{y}$, and one can show that the the correlation between $y$ and the LOO version of $\bar{y}$ is -1!} 

Figure~\ref{fig:chris} shows an example from our ``medium-SNR'' setting with $n=300, p=1000$.
The correlations between $y$ and each feature needs to be larger than about 0.08 in order for the correlation between $y$ and each  LOO feature to also be positive. 

\begin{figure}
\begin{center}
  \includegraphics[width=3in]{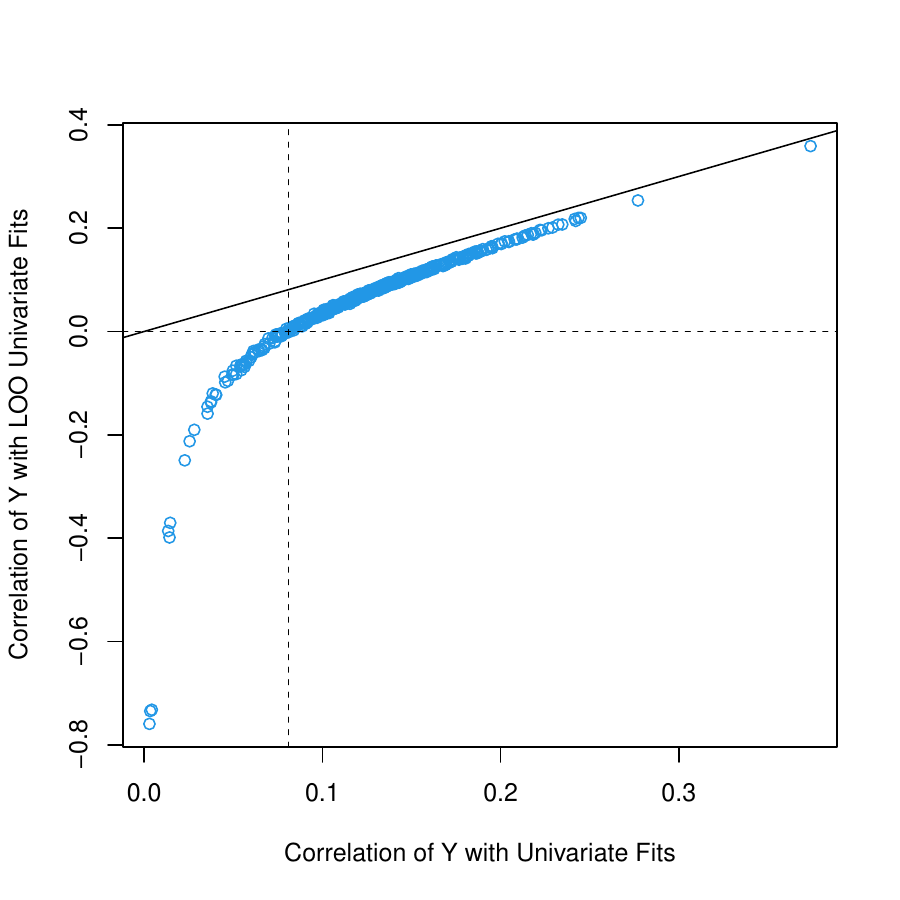}
  \end{center}
  \caption{\em Medium-SNR setting: correlation between the response and the  univariate LOO features, versus 
  the correlation between the response and the  usual  univariate
  (non-LOO) features. The solid line represents equality.}
  \label{fig:chris}
  \end{figure}
 In Appendix~\ref{sec:deriv-supp-sect} we show that if
 the absolute correlation between $y$ and a feature is smaller than
 approximately $\sqrt{2/n}$, then the correlation with the LOO fit
 will be negative. Since we have $n=300$ in this example, this
 evaluates to $0.082$.  Asymptotically, this means that the two-sided $p$-value for
 testing $\rho=0$ has to greater than about 0.16.

 In the second step of the uniLasso algorithm in
 Section~\ref{sec:our-proposal}, we fit a multivariate lasso model
 using all the LOO features, but with a positivity constraint on their
 coefficients. So even though marginally some LOO features are
 negatively correlated with the response, this does not guarantee that
 they will be omitted; however,  in practice they do tend to be omitted.

\section{Application to the   car prices data}
\label{sec:car}
This dataset consists of 25 features for predicting car prices for 205
cars, taken from
\url{https://www.kaggle.com/datasets/hellbuoy/car-price-prediction}
Make of car (22 levels) was first fit to the data and we modeled the residual on the remaining 26 predictors, 11 of which were categorical (and one-hot encoded).

\begin{table}[hbt]
\centering
\begin{tabular}{rrrrr}
  \hline
 & Univariate & Lasso & UniLasso & Polish \\ 
  \hline
body.style-3 & -4.768 & -0.224 & 0.000 & 0.000 \\ 
  wheel.base & 0.781 & 0.035 & 0.000 & 0.000 \\ 
  width & 3.082 & 0.475 & 0.378 & 0.000 \\ 
  height & 0.326 & -0.154 & 0.000 & 0.000 \\ 
  curb.weight & 0.014 & 0.004 & 0.002 & \tblue{-0.001} \\ 
  engine.type-7 & -0.652 & \tblue{3.303}& 0.000 & 0.000 \\ 
  num.of.cylinders-2 & 9.867 & \tblue{-2.546} & 0.000 & 0.000 \\ 
  num.of.cylinders-3 & -13.854 & 0.000 & -0.545 & 0.000 \\ 
  num.of.cylinders-4 & 10.543 & \tblue{-0.667} & 0.000 & 0.000 \\ 
  num.of.cylinders-7 & -0.652 & \tblue{0.009}& 0.000 & 0.000 \\ 
  engine.size & 0.167 & 0.064 & 0.048 & 0.000 \\ 
  fuel.system-2 & -9.143 & \tblue{ 0.208} & 0.000 & 0.000 \\ 
  fuel.system-5 & -0.694 & -1.248 & 0.000 & 0.000 \\ 
  bore & 17.571 &\tblue{ -3.729} & 0.000 & 0.000 \\ 
  horsepower & 0.169 & 0.017 & 0.042 & 0.000 \\ 
  peak.rpm & -0.001 & 0.000 & 0.000 &   \tblue{0.001}\\ 
   \hline
\end{tabular}
\caption{\em Coefficients from various models fit to the car-price data. Coefficients with sign changes relative to the univariate fits are marked in blue.}
\label{tab:car}
\end{table}

Table~\ref{tab:car} shows the univariate  least squares coefficients on the left, and the lasso and uniLasso coefficients in the middle columns.
The right column shows the ``uniLasso polish''  estimates described in Section \ref{sec:polish}. They result from a post-processing of uniLasso in which the lasso is
applied to the uniLasso residuals. Only features having a non-zero coefficient in at least one of the three rightmost columns are shown.
We see that uniLasso produces a much sparser model than the lasso.

For  Figure~\ref{fig:carbox} we took 50 random $2/3- 1/3$  train-test splits, and computed 4 summaries of model performance. 
We see that the test errors  of all methods are similar, while the support size and number of sign-violations are quite different, as expected. The stability figure is most interesting:
for each of the ${50\choose 2}$ pairs of  models, we defined the ``stability'' as ratio of the number of chosen features in common, divided by the
number of unique features in the union of the pair. We see a clear advantage for uniLasso over lasso, where the median proportion of common features
is over 80\%, versus 60\% for lasso. The uniLasso ``polish'', described in Section \ref{sec:polish}  gives a further improvement in stability, while increasing the support
and number of sign violations.

\begin{figure}[hbt]
  \begin{center}
  \includegraphics[width=.8\textwidth]{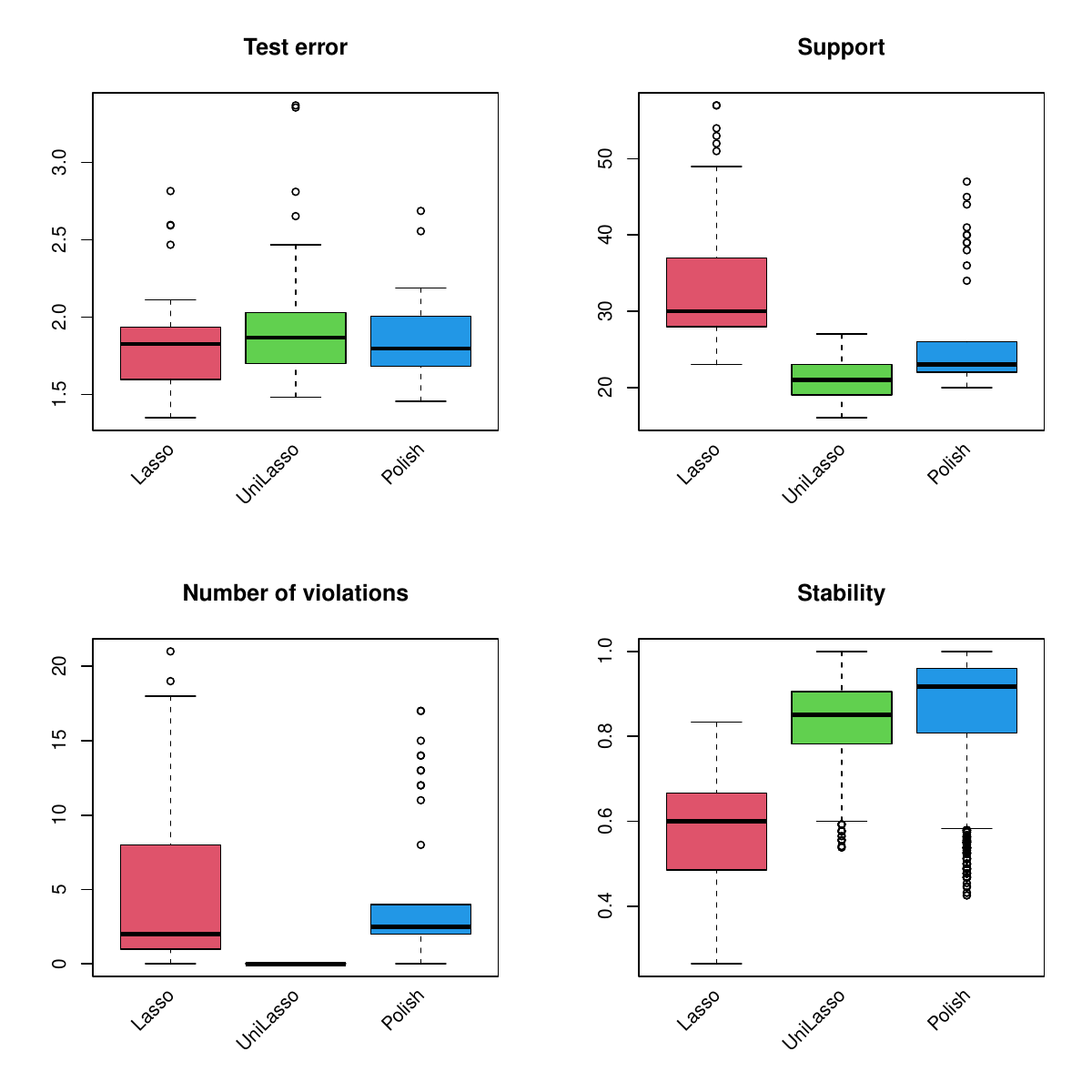}
  \end{center}
  \caption{\em Results for car prices data over 50 train-test splits.
  Shown are test set error, size of the chosen model (support), number
  of sign change violations relative to the univariate signs, and the
  stability --- the average proportion of
   common features over all  model pairs  for each method.}
  \label{fig:carbox}
  \end{figure}

\section{Analysis  of uniLasso with orthogonal features}
\label{sec:ortho}
In this section we derive explicit formula for the uniLasso
coefficients in the special case of orthonormal features.
Figure~\ref{fig:figorth} shows the lasso and uniLasso paths for a simulated example with an orthonormal feature matrix.
They look somewhat different, with the uniLasso path being sparser at
any stage along the path (as measured by the  $\ell_1$ norm of the coefficients).
\begin{figure}[hbtp]
  \includegraphics[width=\textwidth]{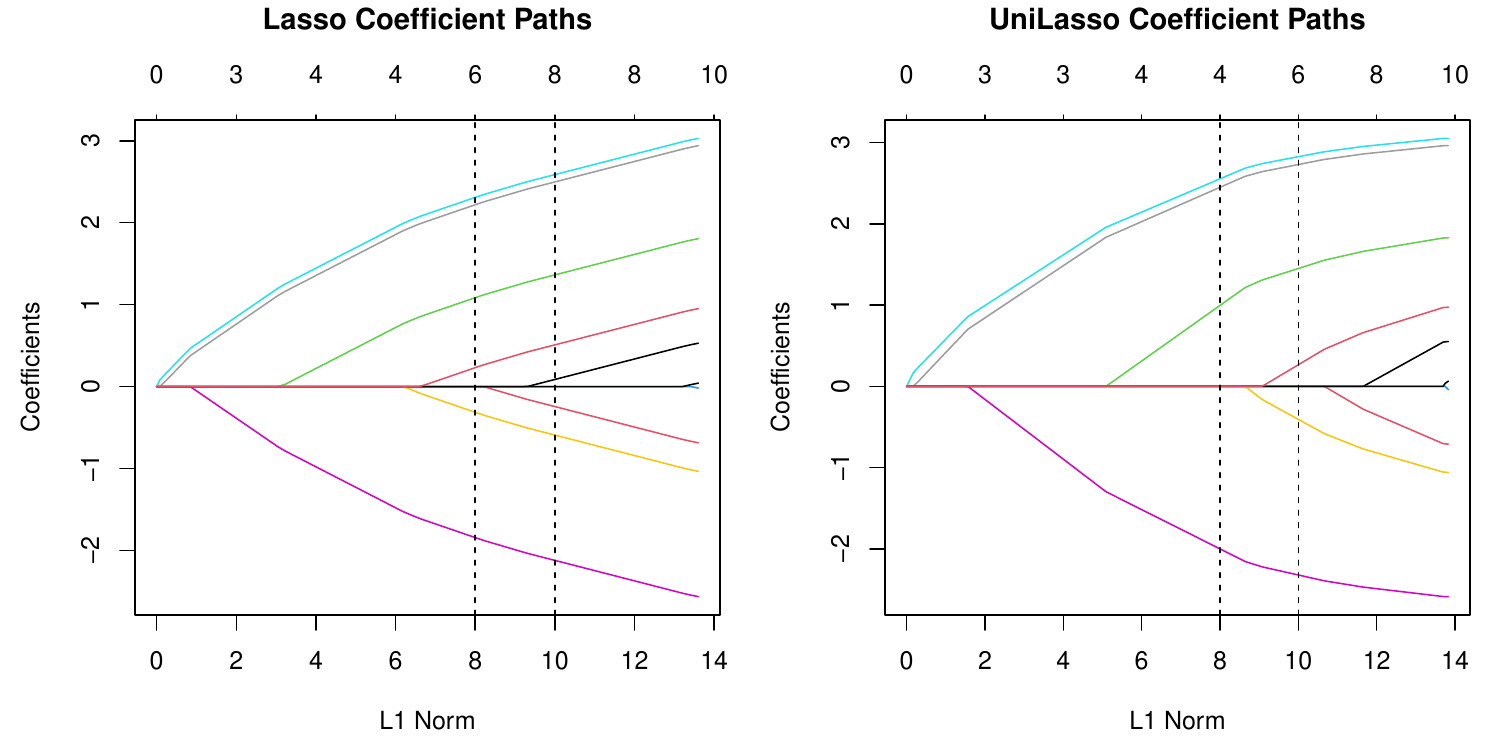}
  \caption{\em Coefficient paths for lasso and uniLasso, in  a
    simulated example with 10 orthonormal features. The end of the
    path represents the least squares fit, and in this case also the
    univariate coefficients. The paths for features with large
    absolute univariate coefficients look very similar in both plots. Those with small
    univariate features have a delayed entry in the right plot.}
  \label{fig:figorth}
  \end{figure}
  
 It is well-known that in this setting, the lasso coefficients are simply soft-thresholded versions of the univariate least squares
 estimates. This is (approximately) true of uniLasso, but with a different thresholding function, as we now show. 
 
Suppose \( X \) is orthonormal so that each column \( x_j \) satisfies \( \|x_j\|_2^2 = 1 \) and \( x_i^T x_j = 0 \) for \( i \neq j \).
Assume also that \( \bar{y}=0\) and \( \bar{x}_j= 0 \)  for each
$j$. Note that the least squares coefficients are
$\hat{\beta}_{0j}=0$ and $\hat{\beta}_{j}=x_j^T y$.

Here we use the actual fits $\hat\eta_j^i$ rather than the LOO fits
$\hat\eta_j^{-i}$ in the second stage --- an approximation that
simplifies the derivation, and often gives very similar solutions to uniLasso.
In this case we have $\hat\eta_j^i=\hat\beta_j x_{ij}$, since $\hat\beta_{0j}=0$.
We can ignore \( \theta_0 \), which will be zero since \( \bar{y} = 0 \) and all the
$\hat\eta_j=\hat\beta_j x_{j}$ have means zero.

The coefficients \( \hat{\theta}_j \) are determined by solving:

\[
\min_{\theta_j}\left\{ \frac{1}{2} \| y - \sum_{\ell=1}^p \theta_\ell {\hat\beta_\ell} x_\ell \|_2^2 + \lambda \| \theta \|_1\right\}
\]

We have the derivative w.r.t. \( \theta_j \):

\[
-{\hat\beta_j} x_j^\top \left( y - \sum_{\ell=1}^p \theta_\ell {\hat\beta_\ell} x_\ell \right) + \lambda \cdot \text{sign}(\theta_j) = 0
\]

Since the $x_j$ are orthogonal and unit norm, and using $x_j^\top y=\hat\beta_j$, we get

$
{\hat\beta_j}^2 \theta_j = {\hat\beta_j^2} - \lambda \cdot
\text{sign}(\theta_j),
$
and hence
\[
\theta_j = \left( 1 - \frac{\lambda}{{\hat\beta_j}^2} \right)_+.
\]

Hence the final coefficients for $x_j$ are
\begin{equation}
  \label{eq:orth}
\hat\gamma_j=\hat\theta_j\cdot\hat\beta_j={\hat\beta_j} \left( 1 -
  \frac{\lambda}{{\hat\beta_j}^2} \right)_+=\mbox{sign}(\hat\beta_j)\left(|\hat\beta_j|-\frac{\lambda}{|\hat\beta_j|}\right)_+
\end{equation}
The last expression can be compared with the similar expression for
the lasso in this situation:
\begin{equation}
  \label{eq:orth2}
\mbox{sign}(\hat\beta_j)\left(|\hat\beta_j|-\lambda \right)_+.
\end{equation}
Figure~\ref{fig:figshrink} shows the shrinkage  functions for ridge regression, lasso, and  uniLasso.
Ridge uses proportional shrinkage, while lasso translates all  coefficient to zero by the same amount.
uniLasso is similar to lasso, except that  the larger coefficients are shrunk less than the smaller ones.

\begin{figure}
 \includegraphics[width=6in]{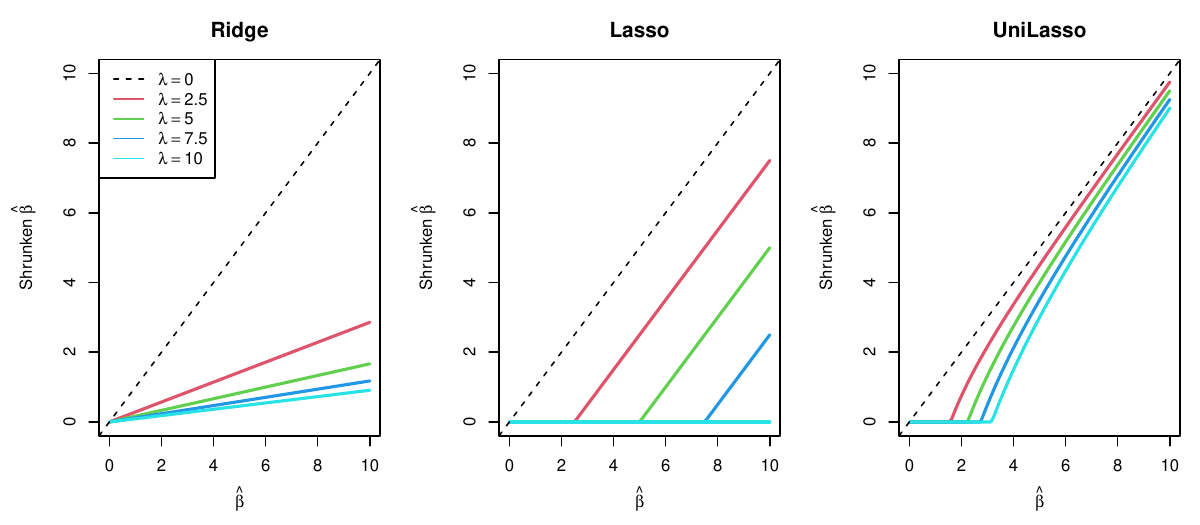}
  \caption{\em Shrinkage functions for ridge regression, lasso and  uniLasso, in  a simulated example with 10 orthonormal features.} 
  \label{fig:figshrink}
 \end{figure}
  
  The SparseNet procedure \citep{MFH2010} uses a thresholding function
  that has a roughly similar shape to that of uniLasso, but its objective is not
  convex. This shrinkage pattern is somewhere between lasso and
  $\ell_0$ or best subset selection.  UniLasso achieves this without
  losing convexity, a significant computational advantage.

\section{Theoretical analysis}
\label{sec:theory}
In this section we study the support recovery and mean-squared error properties of uniLasso.
Let $X_1,\ldots,X_p$ be square-integrable random variables with nonzero variance, defined on the same probability space, and let 
\[
Y = \gamma_0 + \sum_{j\in S} \gamma_j X_j + \epsilon,
\]
where $S$ is a subset of $\{1,\ldots,p\}$, $\epsilon$ is a mean zero random variable that is independent of the $X_j$'s, and the $(\gamma_j)_{j\in S}$ are nonzero coefficients. We will refer to $S$ as the {\it support}. Assume that $Y$ also has nonzero variance. Our data consists of i.i.d.~random vectors $(Y_i, X_{i,1},\ldots,X_{i,p})$, $i=1,\ldots,n$ (where $n\ge 2$), each having the same distribution as $(Y,X_1,\ldots,X_p)$.

The uniLasso algorithm with penalty parameter $\lambda >0$ goes as follows. For each $1\le i\le n$ and $1\le j\le p$, let 
\[
 \hat\beta^{-i}_j := \frac{\frac{1}{n-1}\sum_{k\ne i}Y_k X_{k,j}-(\frac{1}{n-1} \sum_{k\ne i}Y_k)(\frac{1}{n-1} \sum_{k\ne i} X_{k,j})}{\frac{1}{n-1}\sum_{k\ne i}X_{k,j}^2 - (\frac{1}{n-1}\sum_{k\ne i}X_{k,j})^2}
\]
be the regression coefficient from the univariate regression of $Y$ on $X_j$ omitting observation $i$, and let
\[
\hat{\alpha}^{-i}_j := \frac{1}{n-1}\sum_{k\ne i} Y_k - \frac{ \hat\beta^{-i}_j}{n-1}\sum_{k\ne i} X_{k,j}
\]
be the intercept term. Note that $ \hat\beta^{-i}_j$ is a consistent estimate of 
\[
\beta_j := \frac{\cov(Y, X_j)}{\var(X_j)}, % = \mathrm{Corr}(\tY, \tX_j) = \E(\tY\tX_j),
\]
and $\hat{\alpha}^{-i}_j$ is a consistent estimate of
\[
\alpha_j := \E(Y) -\beta_j \E(X_j).
\]
%where 
%\[
%\tY := \frac{Y-\E(Y)}{\sqrt{\var(Y)}}, \ \ \ \tX_j := \frac{X_j - \E(X_j)}{\sqrt{\var(X_j)}}
%\]
%are the standardized versions of $Y$ and $X_j$. 
Then, let 
\[ 
\hat{Y}_{i,j} := \hat{\alpha}^{-i}_j +  \hat\beta^{-i}_j X_{i,j}
\]
be the predicted value of $Y_i$ from this univariate regression. Next, obtain the estimates $\hat{\theta}_j$, $j=0,\ldots,p$, by minimizing
\[
L(\theta_1,\ldots,\theta_p) = \frac{1}{n}\sum_{i=1}^n (Y_i - \theta_0 - \theta_1\hat{Y}_{i,1}- \cdots- \theta_p\hat{Y}_{i,p})^2 +\lambda \sum_{j=1}^p \theta_j
\]
subject to the constraint that $\theta_j\ge 0$ for each $1\le j\le p$. Finally, define 
\[
\hat{\gamma}_j :=\hat{\theta}_j \hat{\beta}_j, % \frac{\hat{\theta}_j}{n}\sum_{i=1}^n  \hat\beta^{-i}_j
\]
to be the uniLasso estimate of $\gamma_j$ for $1\le j\le p$, where
\[
\hat{\beta}_j := \frac{\frac{1}{n}\sum_{k=1}^nY_k X_{k,j}-(\frac{1}{n} \sum_{k=1}^nY_k)(\frac{1}{n} \sum_{k=1}^n X_{k,j})}{\frac{1}{n}\sum_{k=1}^nX_{k,j}^2 - (\frac{1}{n}\sum_{k=1}^nX_{k,j})^2}
\]
is the regression coefficient from the univariate regression of $Y$ on $X_j$, and let 
\[
\hat{\gamma}_0 := \hat{\theta}_0 + \sum_{j=1}^p\hat{\theta}_j \hat{\alpha}_j,
\]
where
\[
\hat{\alpha}_{j} := \frac{1}{n}\sum_{k=1}^n Y_k - \frac{\hat{\beta}_{j}}{n}\sum_{k=1}^n X_{k,j}
\]
is the intercept term from the univariate regression of $Y$ on $X_j$.

The following theorem shows, roughly speaking, that if (1) $\sign(\gamma_j)=\sign(\beta_j)$ for each $j\in S$, (2) the penalty parameter $\lambda$ is bigger than $|\beta_j|$ for all $j\notin S\cup \{0\}$, and (3) both $\log p$ and $\log n$ are small compared to $n\lambda^2$, then with high probability, $\hat{\gamma}_j = 0$ for all $j\notin S\cup \{0\}$ and $\hat{\gamma}_j = \gamma_j + O(\lambda)$ for all $j\in S\cup \{0\}$.
\begin{thm}\label{mainthm}
Suppose that:
\begin{enumerate}
\item $\gamma_j\beta_j> 0$ for each $j\in S$.
\item The covariance matrix of $(X_j)_{j\in S}$ is nonsingular with minimum eigenvalue $\eta$.
\item There is a positive constant $C_0$ such that $\var(Y)\ge C_0$ and $\var(X_j)\ge C_0$ for each $j\in S$.
\item There are positive constants $C_1$ and $C_2$ such that for each $t\ge 0$ and $1\le j\le p$, $\P(|Y|\ge t)$, $\P(|\epsilon|\ge t)$ and $\P(|X_j|\ge t)$ are bounded above by $C_1e^{-C_2 t^2}$.
\end{enumerate}
Let $M_1 := \max_{j\in S\cup\{0\} } |\gamma_j|$, $M_2 := \min_{j\in S} |\beta_j|$ and $M_3 := \max_{j\in S}|\beta_j|$.  Then there are positive constants $K_1,K_2,K_3,K_4, K_5$ depending only on $C_0, C_1, C_2, \eta, M_1, M_2, M_3$ and $|S|$ such that if
\[
K_1\max_{j\notin S\cup \{0\}} |\beta_j | \le \lambda \le K_2, % \min_{j\in S} |\beta_j|,
\]
then %for any $t\ge 0$, %for any $\delta >0$,
\begin{align*}
&\P(\hat{\gamma}_j = 0 \textup{ for all } j\notin S\cup \{0\} \textup{ and } |\hat{\gamma}_j-\gamma_j|\le K_3\lambda   \textup{ for all } j \in S\cup \{0\})\\
&\ge 1 - K_4pne^{-K_5n\lambda^2}.
\end{align*}
\end{thm}
The above theorem is roughly comparable to the available results for the lasso. A close comparison would be, for instance, \cite[Theorem 11.3]{hastieetal15}. Like Theorem \ref{mainthm}, this theorem also assumes a lower bound on the covariance matrix of the covariates in the support, and the maximum difference between $\hat{\gamma}_j$ and $\gamma_j$ for $j\in S$ is of order $\lambda$. However, there is one key difference. The results about lasso, including the one cited above, require a condition known as {\it mutual incoherence} or {\it irrepresentability}. Roughly speaking, it means that if we regress $X_k$ for some $k\notin S$ on $(X_j)_{j\in S}$, the regression coefficients should be small. Notably, our Theorem \ref{mainthm} requires no such relation to hold between the covariates inside and outside the support. All we need is that the univariate regression coefficients of $Y$ on covariates outside the support are small.

Having said that, we make clear that assumption (1) above is a crucial one: namely that  $\gamma_j$ and $\beta_j$ have the same sign for each $j\in S$. The next result gives a natural sufficient condition under which this holds.
\begin{thm}\label{condthm}
Let $\delta_{j,k}$ denote the population value of the univariate regression of $X_j$ on $X_k$. Suppose that for every pair of distinct indices  $j,k\in S$, $\delta_{j,k}\ge 0$ if $\gamma_j$ and $\gamma_k$ have the same sign, and $\delta_{j,k}\le 0$ if $\gamma_j$ and $\gamma_k$ have opposite signs. Then $\beta_j$ is nonzero and has the same sign as $\gamma_j$ for each $j\in S$. Moreover, $|\beta_j|\ge |\gamma_j|$ for each $j\in S$.
\end{thm}

The last result, partly  due to Ryan Tibshirani, gives a necessary and sufficient condition  for equality of signs.

\begin{thm}\label{condthm2}
Suppose that $Y = \sum_{j\in S} \gamma_j X_j +\epsilon$, where all the $\gamma_j$'s are nonzero, and $\epsilon$ has zero mean, finite variance, and is independent of $(X_j)_{j\in S}$. Let $\beta_j$ be the coefficient of $X_j$ in the univariate (population) regression of $Y$ on $X_j$. Let $\delta_{kj}$ be the coefficient of $X_j$ in the univariate (population) regression of $X_k$ on $X_j$. For each $j$, let $A_j$ be the set of $k\in S$ for which $\sign(\delta_{kj}) = \sign (\gamma_k \gamma_j)$ or $\delta_{kj}=0$.  Then $\beta_j\gamma_j \ge 0$ if and only if 
\[
\sum_{k\notin A_j} |\gamma_k \delta_{kj}|\le \sum_{k\in A_j} |\gamma_k \delta_{kj}|. 
\]
In particular, if $A_j=S$, then this holds.
\end{thm}

The proofs of these results are in the Appendix.

\vskip -2in

\section{Some simulation results}
\label{sec:sim}

Figure~\ref{fig:simplots}  shows the results of a comparative study of uniLasso with lasso and three other sparse regression methods.
Shown are  means and standard errors  over 100 simulations.
The methods are:
\begin{itemize}
\item Lasso (using {\tt glmnet}).
\item  UniLasso
\item Polish: a post-processing of the uniLasso solution, described in Section \ref{sec:polish}.
\item Adaptive lasso using penalty weights $1/|\hat\beta_j|$ using the
  univariate coefficients $\hat\beta_j$. These first three methods use  CV to estimate the value of $\lambda$ that minimizes test error.
\item Matching: a variant of the lasso, where we increase the $\lambda$ parameter from the CV-based optimal value, until the support size matches that of uniLasso.
\end{itemize}

The simulation scenarios are:
\begin{description}
\item{1. {\bf Low, medium and high SNR}}:  Here $n=300$, $p=1000$,  $N(0,1)$ features with pairwise correlation 0.5.
The nonzero regression coefficients  are $N(0,1)$; Gaussian errors with SD $\sigma$ chosen to produce low $<1$, medium ($\approx 1$) , or high ($>2)$ SNRs.
\item{2. {\bf Homecourt}}: Here $n=100$, $p=30$ with details as specified in Section~\ref{sec:twoexamples}.
\item{3. {\bf Two-class}}:   We set $n = 200$, $p = 500$ with a binary target $y$. The feature  covariance within each class is AR(1) with $\rho=0.8$,
and the first 20 features are shifted in the $y = 1$ class by 0.5 units.
\item{4. {\bf Counter-example}}: $n=100$, $p=20$, $x_1 \sim N(0,1)$,
  $x_2=x_1+N(0,1)$, $\beta=(1,-.5, 0,0,\ldots 0)$, error SD = 0.5.
\end{description}
In the first three scenarios, 10\%  of the  true coefficients are non-zero; for example, in (1), 100 of the 1000 true coefficients are non-zero.
 Figure~\ref{fig:simplots} shows the test set error  (relative to the lasso) and support size (number of non-zero coefficients).

\begin{figure}[tbp]
\includegraphics[width=6in]{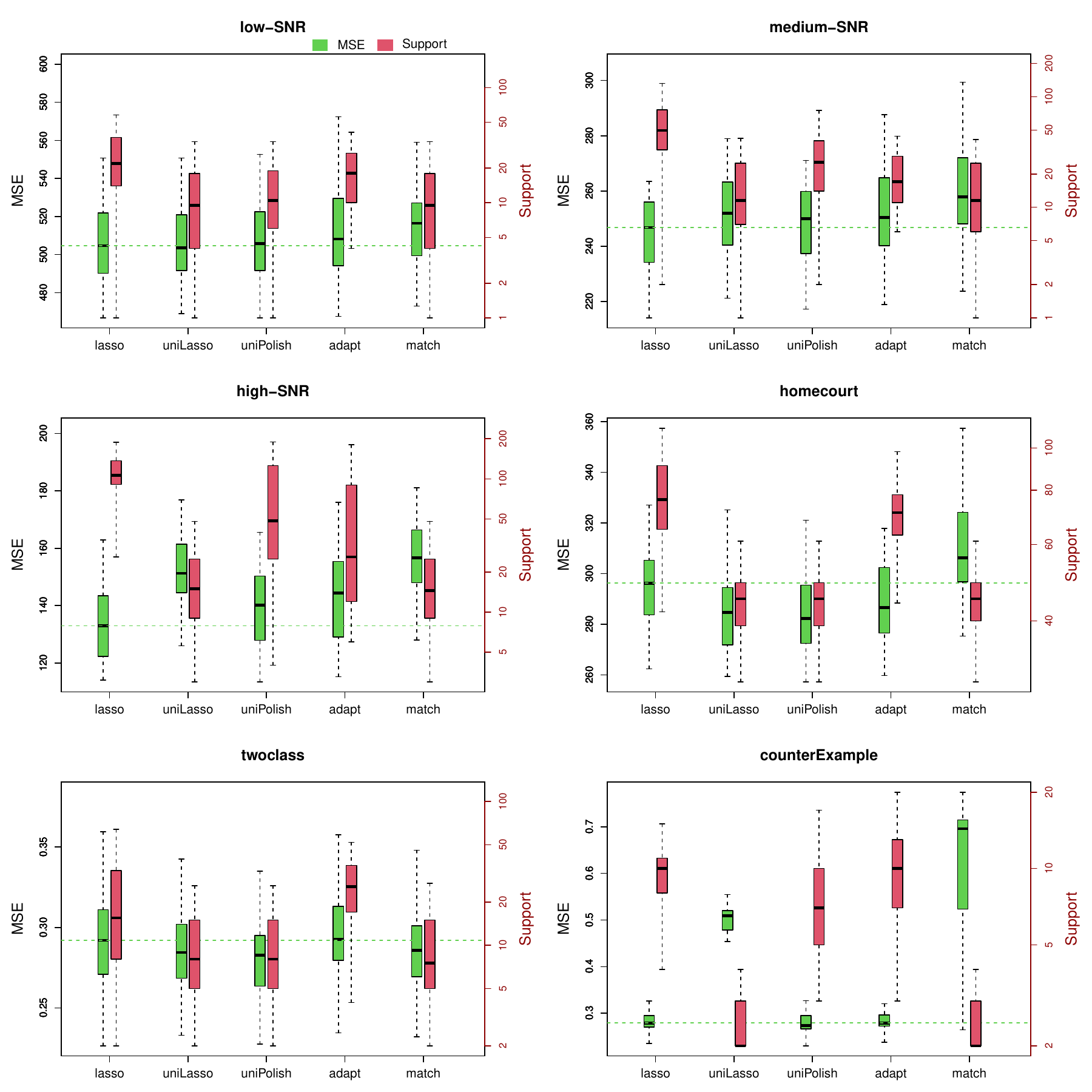}
\caption{\em $N=300$, $p=1000$: test error (Gaussian) or  misclassification rate (Binomial), along with  support size.}\label{fig:simplots}
\end{figure}

Overall we see that uniLasso shows MSE similar to that of lasso, with often a much smaller support. The same general pattern holds with the real datasets.
Exceptions to this are the high-SNR setting
where lasso wins handily and the homecourt setting designed to exploit
uniLasso's strength (recall Theorem~\ref{condthm}).  The ``counter-example'' setting is the classic case where features have positive correlation
but the true regression coefficients have opposite signs;  not
surprisingly uniLasso does poorly. The polish method and adaptive
lasso do reasonably well, while the matching method does poorly.

Of course in practice one has cross-validation to help determine which
method is preferred in a given example. 

Figure~\ref{fig:simplots2} shows the corresponding results for the
simulated examples with $N=300$ and $p=100$.
\begin{figure}[H]
\includegraphics[width=\textwidth]{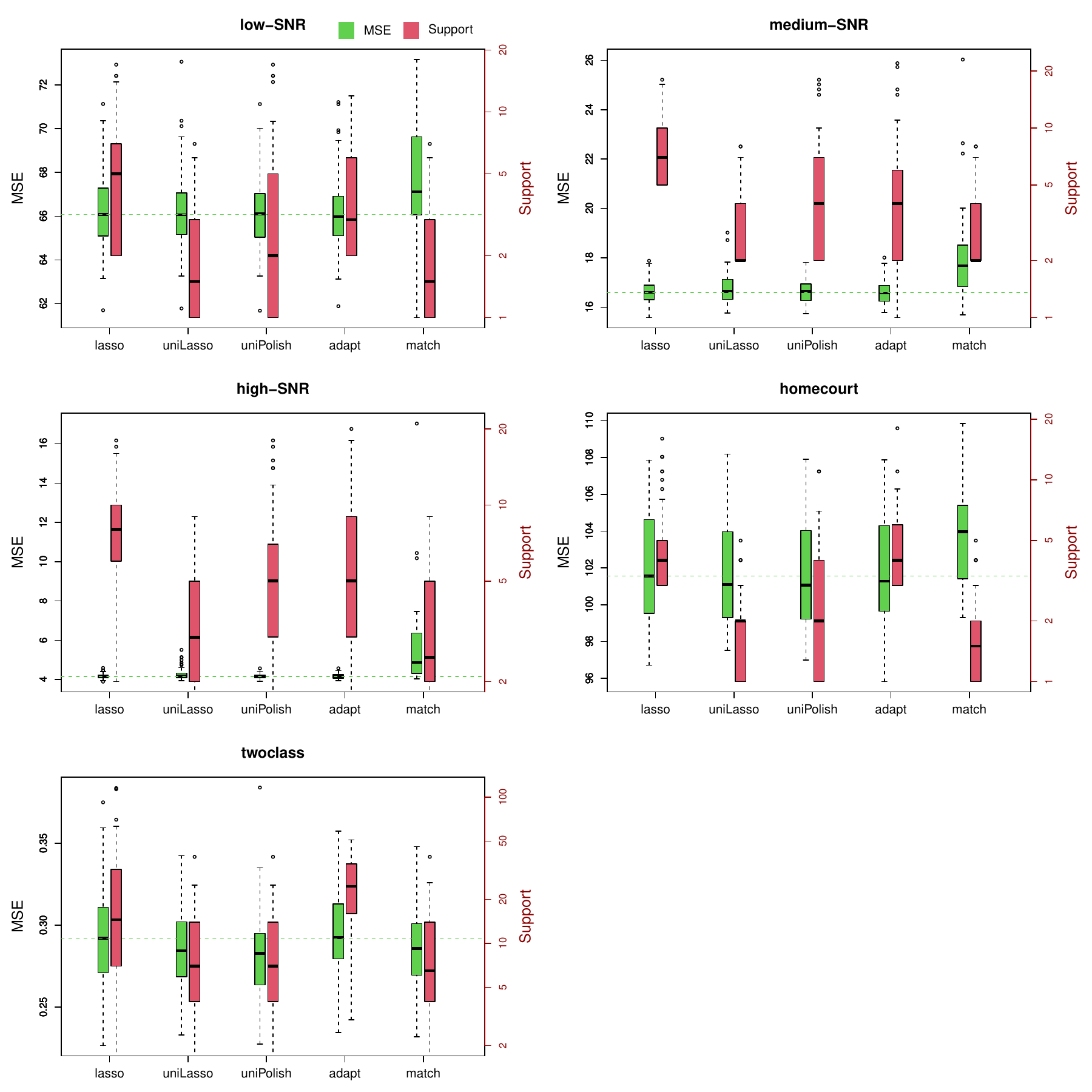}
\caption{ \em $N=300$$ p=100$: Test error (Gaussian) or misclassification rate (Binomial), along with  support size.}
\label{fig:simplots2}
\end{figure}

\section{The unregularized setting: uniReg}
\label{sec:unireg}
Here we consider the case where $\lambda=0$, so that there is no $\ell_1$ penalty,
as detailed in equation (\ref{eqn:unireg}). While the main use case has $n>p$, we note that uniReg can be defined for $p>n$ by taking $\lambda\rightarrow 0$ in uniLasso.
Computationally, this will approximately give the solution with
minimum $\ell_1$ norm, in the same way that the limiting lasso fit gives the least squares solution with minimum $\ell_1$ norm.
Conveniently  even without the $\ell_1$ penalty, the non-negativity constraint  still promotes sparsity in the solution.

In this section we demonstrate the performance of uniReg on simulated
and real data.
\subsection{Simulated data}
In Table~\ref{tab:unireg} 
we compare least squares with uniReg in a subset of the simulated settings defined in Section \ref{sec:sim}.
UniReg performs remarkably well, often winning in MSE and yielding consistently sparser models.

\begin{table}[H]
\begin{small}
\centerline{\bf N=300, p=30}
\centering
\smallskip
\begin{tabular}{lllllll}
  \hline
  Setting & p & SNR & MSE-LS & MSE-uniReg& Supp-LS & Supp-uniReg  \\ 
  \hline
 low-SNR & 30 & 0.092 & 71.339 & 65.171& 30 & 4  \\ 
   medium-SNR & 30 & 0.369 & 17.836 & 16.489& 30 & 4  \\ 
   high-SNR & 30 & 1.475 & 4.459 & 4.178& 30 & 4  \\ 
   homecourt & 30 & 0.059 & 111.771 & 101.542& 30 & 4  \\ 
   \hline
\end{tabular}

\medskip

  \centerline{\bf N=300, p=100}
\smallskip
\begin{tabular}{lllllll}
  \hline
 Setting & p & SNR & MSE-LS & MSE-uniReg& Supp-LS & Supp-uniReg \\ 
  \hline
 low-SNR & 100 & 0.318 & 97.678 & 69.655& 100 & 8 \\ 
   medium-SNR & 100 & 1.272 & 24.428 & 18.559 & 100 & 9 \\ 
   high-SNR & 100 & 5.087 & 6.104 & 5.82& 100 & 7 \\ 
   homecourt & 100 & 0.228 & 149.683 & 104.979& 100 & 13 \\

   \hline
\end{tabular}

\medskip
\centerline{\bf N=300, p=1000}
\smallskip
\begin{tabular}{lllllll}
  \hline
  Setting & p & SNR & MSE-LS & MSE-uniReg& Supp-LS & Supp-uniReg   \\ 
  \hline
 low-SNR & 1000 & 0.578 & 842.191 & 510.837& 300 & 24.5 \\ 
   medium-SNR & 1000 & 1.48 & 374.915 & 249.076& 300 & 27 \\ 
   high-SNR & 1000 & 3.613 & 186.171 & 147.276& 300 & 27 \\ 
   homecourt & 1000 & 1.53 & 577.852 & 420.508& 300 & 135.5 \\ 
   \hline

\end{tabular}
\caption{\em Comparison of least squares with uniReg in a subset of simulated settings defined earlier.}
%\medskip
\label{tab:unireg}
\end{small}
\end{table}

\subsection{Real data}\label{sec:real-data}
We compared least squares and uniReg on the  regression datasets from
the UCI database\footnote{\url{https://archive.ics.uci.edu/datasets}}.
There were 23 datasets in all. Some had multiple targets: we treated each target separately and averaged the results.
We randomly sampled each dataset into a $70/30$ training/test split, and report the test set MSE and support in Figure~\ref{fig:UCI}.
One dataset had a small sample ($n=60$)  which produced a very large MSE from least squares; we omit this from the plot for visibility.
\begin{figure}
\begin{center}
\includegraphics[width=.8\textwidth]{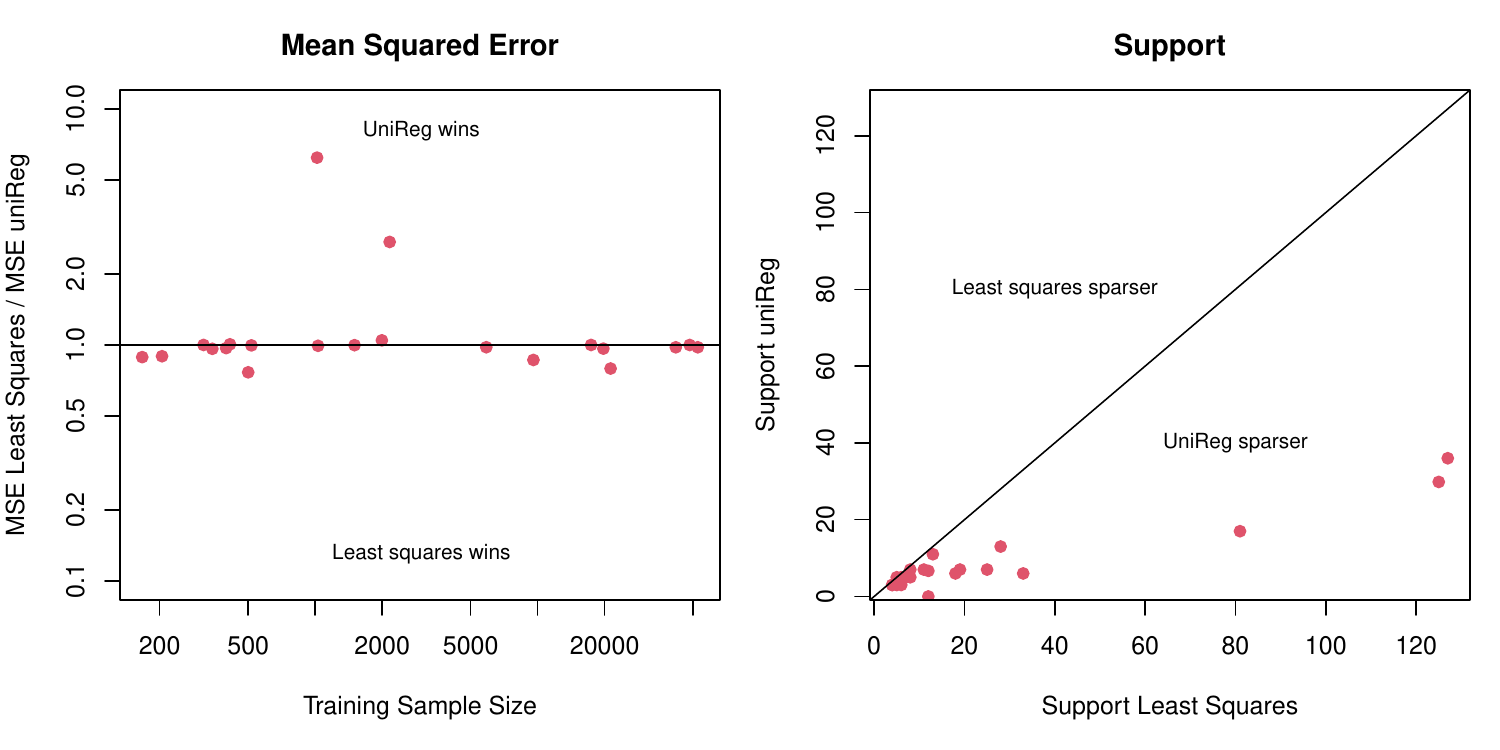}
\end{center}
\caption {\em LS vs uniReg on regression datasets from the UCI database. All datasets have $n>p$. The left plot shows the MSE ratio on the log scale.
The right panel shows the support sizes of the chosen models.}
\label{fig:UCI}
\end{figure}

The two methods perform similarly in MSE; but uniReg has much smaller support, averaging about  $1/3$ that of the least squares support (which is the number of features $p$).

\subsection{Theory for uniReg without LOO}\label{sec:theory-unir-with}
Our data consists of $(X_i, Y_i)$,  $i=1,\ldots,n$, where 
\[
Y_i = \beta_0+ X_i^T\beta +\epsilon_i, 
\]
where $\beta_0\in \R$, $\beta = (\beta_1,\ldots,\beta_p)\in \R^p$, $X_i$ are i.i.d.~$N_p(0,\Sigma)$ random vectors, where $\Sigma$ is a positive definite matrix with minimum eigenvalue $\lambda_0$ and maximum eigenvalue $\lambda_1$,  and $\epsilon_i$ are i.i.d.~$N(0,\sigma^2)$ random variables for some $\sigma>0$. Let $r := p/n$.  We assume that $r \in (0,r_0)$ for some $r_0<1$.

Let $\aols$ and $\bols$ denote the OLS estimates of $\beta_0$ and $\beta$. Let $\aur$ and $\bur$ denote the uniReg estimates of $\beta_0$ and $\beta$ {\it without leave-one-out}, defined as follows. First, compute the univariate coefficients $\buni_j$, $j=1,\ldots,p$, as the coefficient of $X_j$ when $Y$ is regressed solely on $X_j$ with an intercept term. The estimates $\aur$ and $\bur$ are obtained by minimizing
\[
\sum_{i=1}^n (Y_i - a - X_i^T b)^2
\]
over $a\in \R$ and $b\in \R^p$ subject to the constraint that $b_j\buni_j\ge 0$ for $j=1,\ldots,p$. 
Let $\beta^{\mathrm{uni}}_{0,j}$ and $\bbuni_j$ be the population univariate coefficients; that is,
\[
\E(Y_i | X_{i,j}) = \beta^{\mathrm{uni}}_{0,j} +\bbuni_j X_{i,j}.
\]
Define three vectors $\mu, \muols, \muur\in \R^n$ as 
\begin{align*}
&\mu_i := \E(Y_i|X_i) = \beta_0 + X_i^T\beta,\\
&\muols_i :=\aols + X_i^T\bols,\\
&\muur_i := \aur + X_i^T\bur,
\end{align*}
for $i=1,\ldots,n$. Let $q$ be the number of $j$ such that $\beta_j =0$ and $\bbuni_j \ne 0$. The following theorem shows that if $\beta_j$ and $\bbuni_j$ have the same sign for all $j$, then  $\E\|\mu-\muur\|^2$ is less than or equal to $\E\|\mu -\muols\|^2$ minus a constant times $q$. %, plus an exponentially negligible term.
\begin{thm}\label{uniregthm}
Suppose that $\bbuni_j \beta_j \ge 0$ for each $1\le j\le p$. Let $\delta$ be the minimum of $|\bbuni_j|$ over $1\le j\le p$ such that $\bbuni_j\ne 0$.  
Then there are positive constants $C_0, C_1$ depending only on $\lambda_0$, $\lambda_1$, $\delta$, $r_0$, and $\sigma$ such that if $n\ge C_0$, then 
\[
\E\|\mu-\muur\|^2 \le \E\|\mu-\muols\|^2 - C_1 q.
\]
\end{thm}
The theorem is proved in Appendix~\ref{sec:proof-theor-refun}. We have not been able to prove the LOO version of uniReg (that is, the actual version that we have proposed). We leave it as an open question.

\section{Real data examples for uniLasso}
\label{sec:real}
Table \ref{tab:real}  gives a high level summary of  six  datasets that we examined in this study.
Figure  \ref{fig:realplots} shows the results:
these are test error and support size  averaged over 100 train/test  50-50 splits.
\begin{table}[hbt]
\begin{center}
\begin{tabular}{lrrrr}
Dataset & n & p &Feature type  &Outcome type\\
\hline
Spam&2301 & 57 & mixed &binary \\
NRTI& 1005 & 211 &binary mutations &  drug effectiveness (quantitative)\\
Leukemia& 72 & 7129& gene expression (quantitative) & binary\\
DLBCL&240 & 1000 &gene expression (quantitative) &survival\\
Breast Cancer & 157 & 1000& gene expression (quantitative) & survival\\
Ovarian & 190 & 2238 & Mass spec peak intensities & binary\\
\end{tabular}
\end{center}
\caption{\em Summary of datasets used in our study.}
\label{tab:real}
\end{table}

We see the same general trend as in the simulated examples: uniLasso tends to give test error similar to that of lasso, with often smaller support size.
\begin{figure}[hbt]
\includegraphics[width=.9\textwidth]{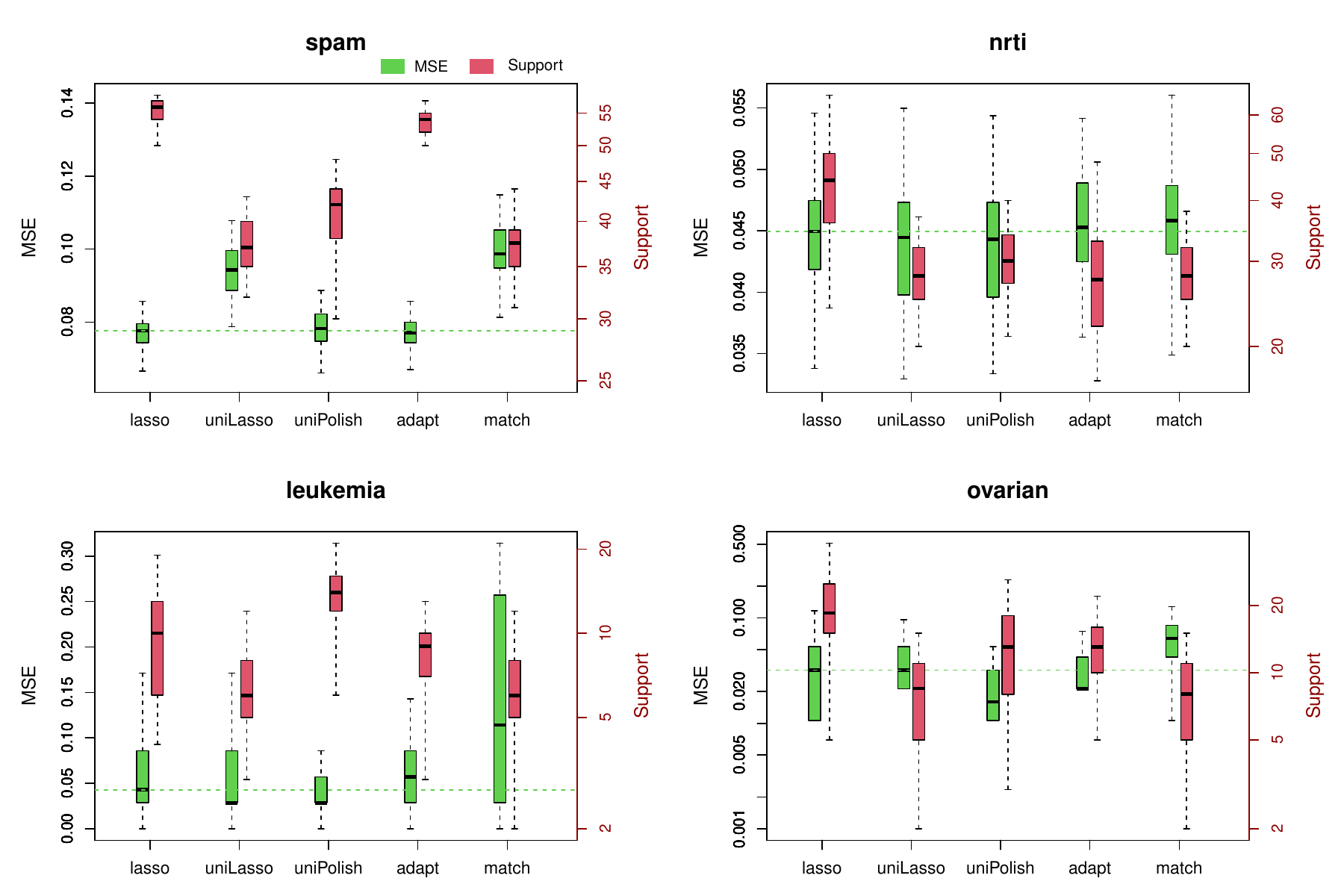}
\caption{\em Real data examples: Blue lines: relative test error (Gaussian)  misclassification rate (Binomial) or deviance (Cox), with support size superimposed in gold.}
\label{fig:realplots}
\end{figure}
\section{Multiclass models}
\label{sec:multiclass}
Here we consider a classification problem with more than two classes.
The multinomial model is natural for the multiclass problem, but does
not generalize easily to uniLasso. The reason is that for each feature
there are coefficients for each class, so it is not clear how to
replace the feature with a LOO version. Furthermore the coefficients
per class are only determined up to a shift,
and hence a non-negativity constraint does not make sense.
Instead we take a ``one-versus rest'' (OVR) approach, where the
uniLasso binary classifier is used to classify each class from the
others.
Hence each class gets a predicted probability, and  
with OVR one classifies to the class with the highest probability as the predicted class.

We applied this to a dataset on childhood cancer, with 63 training and  25 test samples in 4 classes
\citep{Khanetal2001}.
  With the original train/test split we get  zero test errors with \texttt{glmnet}.
We repeated the train/test split 50 times and obtained these results:

\begin{table}[ht]
\centering
\begin{tabular}{lrrrr}
  \hline
 & Ave \# Errors& SD & Ave Support & SD\\ 
  \hline
Lasso multinomial& 0.96 & 0.03 & 18.68 & 0.08 \\ 
  UniLasso OVR & 0.18 & 0.01 & 27.02 & 0.05 \\ 
Lasso OVR  &0.22 &  0.01 &     58.08 & 0.08\\
   \hline
\end{tabular}
\begin{tabular}{lrrrrrr}
  \hline
Number of Errors & 0 & 1 & 2 & 3 & 4 & 7 \\ 
  \hline
Lasso  multinomial&  27 &   9 &   8 &   4 &   1 &   1 \\ 
UniLasso  OVR&  41 &   9 \\ 
Lasso OVR & 39 & 11\\
   \hline
\end{tabular}
\caption{\em Results for the four-class cancer example. Overall
  uniLasso OVR makes the fewest misclassification errors averaged over the
  50 train-test splits.  Of the 50  train(63)/test(25)  splits, 41
  made zero errors on the 25  test data points, and 9 made 1 error,
  averaging 0.18. Lasso OVR is sligthly worse, and lasso multinomial
  makes almost one error on average per train/test split. On the other hand lasso
  multinomial yields the sparsest model, since the other two select
  features separately for each class.
}
\end{table}

% \subsection{A GWAS example}

% I SUGGEST THIS EXAMPLE BE  REMOVED

% \label{sec:GWAS}
% The data here come from the Kaggle site 

% \url{https://www.kaggle.com/datasets/seascape/snp-dataset-for-gwas}

% It is described as a simulated dataset comparable to the Illumina 650K human array, for SNP genotyping.
% There are 1000 individuals, 
% and Genotypes (0, 1 or 2) for 482,906 markers 
% For illustration, we selected the 10,000 markers with highest variance and divided the data into training and test sets each of size 500.
% The target measure is quantitative. 
% \begin{figure}
%  \includegraphics[width=6in]{GWAS.pdf}
%  \label{fig:GWAS}
%  \caption{\em GWAS example: CV and test set R-squared curves}
%  \end{figure}
%  In Figure \ref{fig:GWAS} we see that uniLasso delivers a model with slightly higher test $R^2$ using fewer SNPs.
%  Table \ref{tab:GWAS} summarizes the signs of the lasso and uniLasso coefficients in final chosen models.
%  Interestingly, ignoring the zeroes, there are no sign disagreements!
 
%  \begin{table}[ht]
% \centering
% \begin{tabular}{rrrr}
%   \hline
%  & uni$<$0 & uni=0 & uni$>$0 \\ 
%   \hline
% lasso$<$0 &  33 &  13 &   0 \\ 
%   lasso=0 &   1 & 9901 &   0 \\ 
%   lasso$>$0 &   0 &  13 &  39 \\ 
%    \hline

% \end{tabular}
%    \caption{\em GWAS example: signs of lasso and uniLasso coefficients in final chosen models.}
%  \label{tab:GWAS}
% \end{table}

\section{The use of external univariate scores}
\label{sec:external}

Consider the setting where we have  our training set \texttt{T}  and
also external data \texttt{E}  from the same domain (e.g. disease), We
assume that  \texttt{E}
does not contain  raw data  but only summary results.  Specifically, \texttt{E} contains just univariate coefficients and standard errors for each feature.
This setting occurs fairly often in biomedical settings where investigators are not willing to share their raw data, but do publish and share summary results.

We demonstrate here how uniLasso can make  productive use of external scores. 
The idea is to use  the univariate coefficients  from \texttt{E}, rather than computing the  LOO estimates from the training set. 
Then in step 2 of uniLasso, we proceed as usual.

To investigate this scheme, we generated data \texttt{T}
with $n=300$, $p=1000$, $\mbox{SNR}=1.5$,  with feature  covariance  AR(1), with $\rho=0.8$  and the non-zero regression coefficients distributed as $U(0.5,2)$, and  positioned on every other feature $(1,3,\ldots 99)$.
We also generated 600 extra samples from the same distribution to
serve as an
external data set  \texttt{E}, and in addition a large test set.
Figure~\ref{fig:extra}
shows the test-set MSE over 200 simulations from the following strategies:

\begin{enumerate}
\item  lasso and uniLasso, both applied just to the training set;
\item uniLasso+$m$ (green)  where the univariate scores are derived
  from $m$ extra samples from \texttt{E}, and  are used in place of the LOO estimates;
\item unilLasso (red) on all $300+m$  samples; 
\end{enumerate}

\begin{figure}[hbt]
\begin{center}
\includegraphics[width=\textwidth]{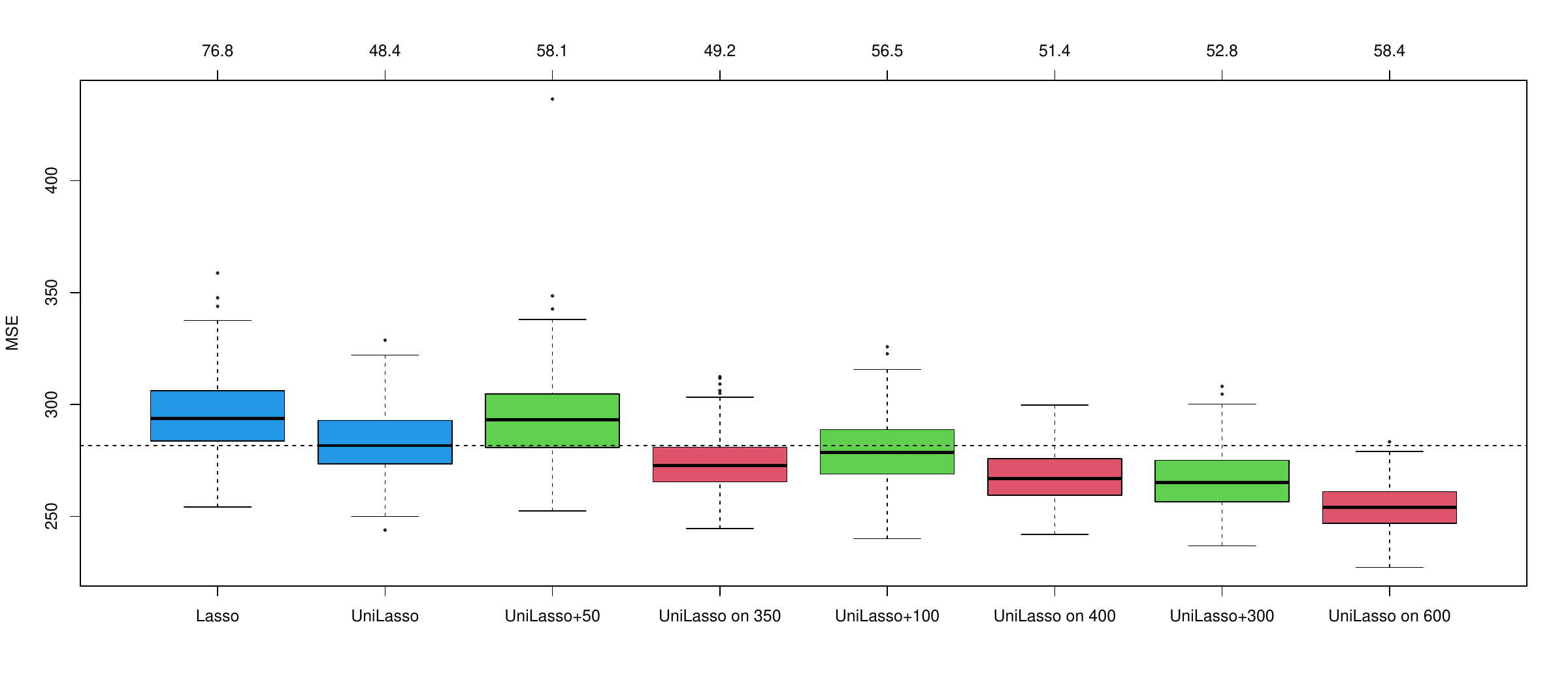}
\end{center}
\caption{\em Results from  an experiment examining the use of external data with uniLasso. The initial training set has 300 observations and we generated 600 additional samples. Shown in blue are the test set MSE for the lasso, and uniLasso,
both applied just to the training set; 
 uniLasso+$m$ (green), where the univariate scores derived from the $m$ extra samples are used in place of the LOO estimates, and unilLasso (red) on $300+m$  samples  for $m=50,100 \ldots  300$.
}
\label{fig:extra}
\end{figure}

We see that adding the scores from just 50 external observations does
slightly worse than uniLasso on the original data, but builds sparser
models than lasso at no cost in test MSE.
If instead, we use the extra 50 samples as in strategy 3, the performance
improves.

If we increase $m$ to 100, both strategies 2 and 3 improve, but uniLasso
using the augmented data performs better each time.

So it appears strategy 2 eventually (as we increase $m$) outperforms lasso and uniLasso
both fit on only the original 300 observations.
But if we have the raw external data, strategy 3 is the clear winner.

This is one simulation scenario, and it is of course possible that in
others the pattern may not be as clear.

\section{A uniLasso polish}
\label{sec:polish}

As we have seen earlier, the uniLasso procedure can sometimes perform  considerably worse than the lasso, especially in cases where its sign constraint 
is too restrictive. An example is the ``counter-example'' problem in Section~\ref{sec:sim} , where two positively correlated features have different
coefficient signs in the population  multivariate model. While problematic, the CV estimate of error will alert the user that uniLasso is not working well in a 
given problem. In this section, we propose a simple ``post-processing'' of the uniLasso solution that can help to diagnose shortcomings of uniLasso
and remedy them.

The idea is very simple: we run the usual lasso (with-cross-validation) to the residual from the uniLasso fit, and then ``stitch'' together the two solutions.
Here is the algorithm in detail, using the \texttt{glmnet} R package terminology:

\begin{enumerate}

\item Run \texttt{cv.uniLasso} on $X, y$ , and extract  the fitted linear predictor $\hat y$ at $\hat\lambda_{min}.$
\item Apply \texttt{cv.glmnet} on $X, y$ , with offset $\hat y$, giving final predictions $\hat y_g$.
\item ``Stitch'' together the two solutions to produce a single path,
  starting at the original uniLasso solution.
\end{enumerate}

Note that in the Gaussian model the use of the offset is equivalent to
applying the lasso to the residual in step 2; expressed as an offset
it allows application of the idea to other GLM families.
\begin{figure}[hbt]
\begin{center}
\includegraphics[width=\textwidth]{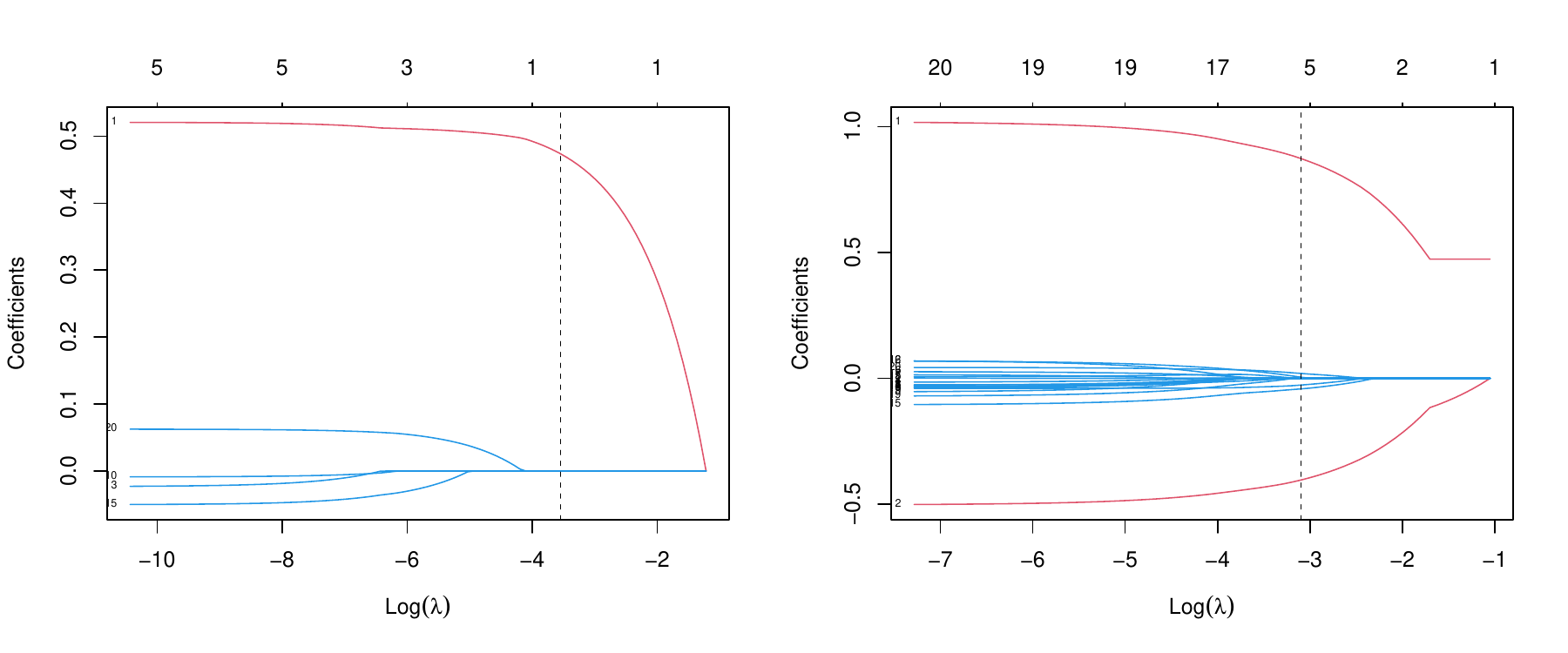}
\caption{\em UniLasso polish applied to the counter-example problem of Section~\ref{sec:sim}.
The left panel shows the {\em\texttt{uniLasso}} solution path, with the dashed
vertical line indicating the solution chosen by cross-validation. The
right panel shows the {\em\texttt{polish.uniLasso}} solution path, which
starts at the {\em\texttt{uniLasso}}, and the vertical dashed line
indicates the model chosen by cross-validation.  The red curves are the
support features $x_1$ and $x_2$.
}
\label{fig:polish}
\end{center}
\end{figure}

Figure~\ref{fig:polish} shows the results of the uniLasso polish applied to the counter-example problem of Section~\ref{sec:sim}.
There are two support features ($x_1$ and  $x_2$). 
In the left panel uniLasso enters just feature $x_1,$ because $x_2$
has  a positive univariate coefficient but a negative multivariate
coefficient. 
% The middle show the post-processing by lasso, where feature 2  is entered, followed by an additional adjustment to the coefficient of feature.
% The two paths are ``stitched'' together to form the right-hand
% panel. The test errors of lasso, uniLasso and uniLasso polish are
% 0.2, 0.51 and  0.31  respectively.
The right panel shows the ``polished'' path. Via the offset it starts
at the uniLasso solution (chosen by CV), and immediately enters $x_2$,
and eventually growing the coefficients of both $x_1$ and $x_2$. In
the process, the model lets in some of the noise variables.

In Section \ref{sec:sim} we included the uniLasso polish in all of the settings, and its performance overall is excellent.

\section{Does CV work in our two stage uniLasso  procedure?}
\label{sec:cv}
In the uniLasso algorithm we use the target $y$ in both steps, but for (significant) computational convenience, we only apply cross-validation in the 2nd step.  Thus one should be concerned
that cross-validation may not perform well here. In the simulations of Section~\ref{sec:sim} we used this form of cross-validation to choose the $\lambda$ tuning parameter, and it seems to have done a 
reasonably good job.

\begin{figure}[hbt]
\includegraphics[width=\textwidth]{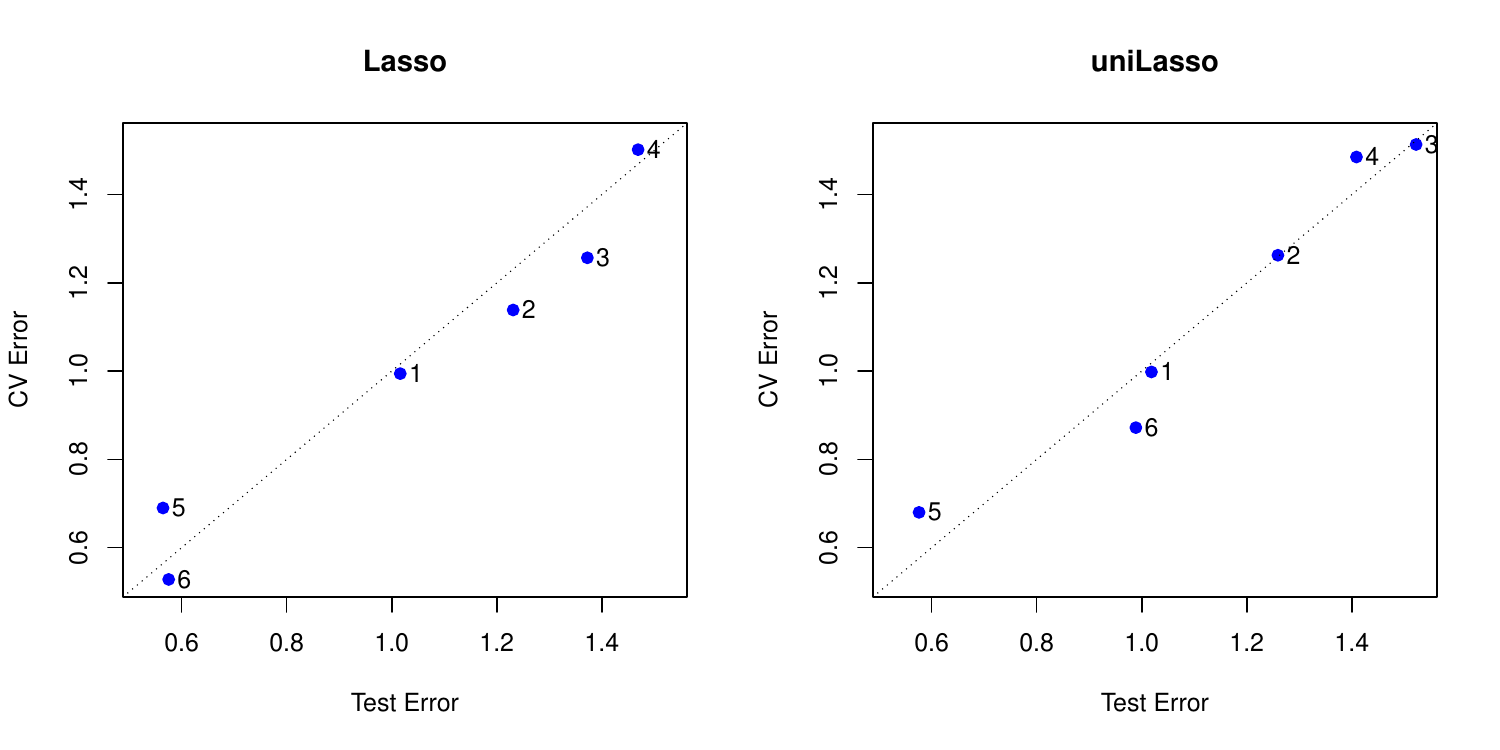}
\caption{\em  CV error at the selected value of $\lambda$ versus the test error, for the lasso (left) and  uniLasso  (right).
Each datapoint represents  one of our simulated datasets with $p>N$
studied earlier.
Both CV estimates track the test error reasonably well.
 }
\label{fig:figcv}
\end{figure}

A related question is whether the error reported by CV is a good estimate of test error, especially at the selected value of $\lambda$.
Figure \ref{fig:figcv} shows the CV error  at the selected value of $\lambda$ versus the test error, for the lasso (left) and  uniLasso  (right).
We see that both CV estimates are quite good. It seems clear that our use of cross-validation in uniLasso is not a major concern.

Finally, a referee asked whether  a reason for uniLasso's
strong performance was its ability to estimate the the CV tuning parameter $\lambda$. 
For the car price data, Figure~\ref{fig:besterrors} shows the ratio of
the achieved test error using CV divided by the minimum test error
over the $\lambda$ path, for each of the 50 train/test splits.
We see that both lasso and uniLasso do an (equally) good job of estimating the optimal $\lambda$. The average ratio  $\lambda_{cv}/\lambda_{opt}$ was 0.87 and 0.38 
for the lasso and uniLasso, respectively.

\begin{figure}
\begin{center}
\includegraphics[width=.4\textwidth]{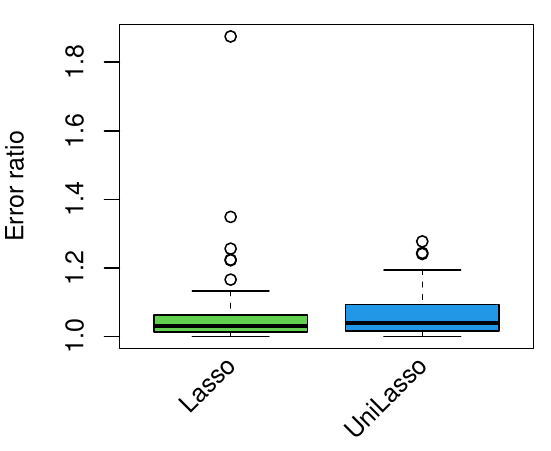}
\end{center}
\caption{\em  Car price data: ratio of the achieved test error using CV divided by the minimum test error over the $\lambda$ path.}
\label{fig:besterrors}
\end{figure}

\section{UniLasso for GLMs and the Cox survival model}
\label{sec:UniLassocox}
Our discussion of uniLasso has focussed on squared-error loss. However,
we can use the same idea in any scenario where we are able to fit a
lasso model and produce a scalar prediction function $\eta(x)$. This
includes in particular all generalized linear models, and the Cox proportional
hazards model. These models are included in \texttt{glmnet}.

The steps are almost the same as before (we will discuss for GLMs):
\begin{enumerate}
\item Fit the $p$ univariate GLMS, and produce the linear predictor
  functions $\hat\eta_j^i = \hat\beta_{0j} +\hat\beta_jx_{ij}$. Compute the LOO predictions
  $\hat{\eta}_j^{-i}$ for the $n$ training 
  observations for each of these.
\item Using these LOO predictions as features, fit a non-negative lasso GLM to the
  response, yielding coefficients $\hat\theta = (\hat{\theta}_0,\hat{\theta}_1,\cdots,\hat{\theta}_p)^\top$.
\item Return the composite estimated linear predictor
  \begin{equation}
    \label{eq:collapse2}
    \hat{\eta}(x) = \hat\theta_0 + \sum_{j=1}^p\hat\theta_j\hat\eta_j(x_j),
  \end{equation}
which collapses as before to a linear model $ \hat\eta(x)=\hat\gamma_0 + \sum_{j=1}^p\hat\gamma_j x_j$.
\end{enumerate}
Step~1 can be computationally challenging. Estimating the $p$
univariate functions $\hat\eta_j(x_j)$ is fairly straightforward,
since each require a few very simple Newton steps. The LOO estimates
are more challenging. While for squared-error loss, we have simple
formulas for computing these exactly, this is not the case for
nonlinear models. However, there are good approximations available
that are computationally manageable. We discuss all the computational
aspects in the next section.

\section{Efficient computations for uniLasso}
\label{sec:effic-comp-unin}

\subsection{Least squares}
\label{sec:least-squares}
We first discuss the computations for squared-error loss. For feature
$j$ we have to fit a univariate linear model, and then compute its LOO
predictions. We are required to fit the model
$\eta_j(x_j) = \beta_{0j}+\beta_jx_j$ using the data
$\{x_{ij},y_i\}_1^n$. This task is simplified if we standardized $x_{ij}$
to have mean zero and unit variance: $z_{ij} = (x_{ij}-\bar{x_j})/s_j$,
where $\bar{x}_j=\frac1n\sum_{i=1}^nx_{ij}$, and
$s_j^2=\frac1n\sum_{i=1}^n(x_{ij}-\bar{x}_j)^2$.
Then the least squares estimates using the $z_{ij}$ are
$(\hat{\delta}_{0j},\;\hat\delta_j)=(\bar y,\;
\frac1n\sum_{i=1}^nz_{ij}y_i)$. These are mapped back to least squares
estimates for $x_{ij}$ via
$(\hat{\beta}_{0j},\;\hat\beta_j)=(\hat{\delta}_{oj}-\hat{\delta_j}\bar{x}_j/s_j,\;\hat\delta_j/s_j)$. 

This standardized form is particularly useful for computing the LOO
fits, since the fits are invariant under affine transformations of the
features (so we can use the $z_{ij}$ instead of the $x_{ij}$). We have the classical formula for the LOO residual in linear
regression
\begin{equation}
  y_i - \hat\eta_j^{-i} = \frac{y_i-\hat\eta_j^i}{1-H_{ii}},
  \label{eqn:LOOres}
\end{equation}
where $H_{ii}$  is the $i$th diagonal entry of the \emph{hat} matrix.
It is straightforward to show that using the $z_{ij}$ gives
$H_{ii} = 1/n +z_{ij}^2/n$. Now simple algebra backs out an expression
for $\hat\eta_j^{-i}$.

These operations can be performed efficiently for all $i$ and $p$
simultaneously using matrix
operations in \texttt{R}, without the need for any loops. In
particular we can perform \emph{Hadamard} (elementwise) arithmetic
operations (eg multiplication) of like-sized matrices in single operations.

\subsection{Generalized linear models}
\label{sec:gener-line-models}
The discussion here also applies to the Cox model. Fitting a GLM by
maximum likelihood is typically done using a Newton algorithm. This
iterative algorithm amounts to making a quadratic approximation to the
negative log-likelihood at the current solution, and then solving the
quadratic problem to get the updated solution. For GLMs this can be
cast as {\em iteratively reweighted least squares} (IRLS).

Given the fitted linear predictor  vector $\eta_j^{(\ell)}$ at
  iteration $\ell$,
  one forms a \emph{working} response vector $z_j^{(\ell)}$ that depends
  on $\eta_j^{(\ell)}$ and $y$ and other properties of the particular
  GLM family, and an
  observation weight vector $w_j^{(\ell)}$, and fits an updated model
  $\eta_j^{(\ell+1)}$ by weighted least squares of $z_j^{(\ell)}$ on $x_j$
  with weights $w_j^{(\ell)}$. Usually 4 iterations are sufficient.

  Once again we can fit all the univariate GLMs using matrix
  operations at the same time --- i.e. we perform each IRLS step
  simultaneously for all $p$ univariate GLMS. The expressions are only slightly more
  complex than in the unweighted case.

  What about the LOO fits $\hat\eta_j^{-i}$ for each univariate GLM?
  Here we make an approximation, as recommended by
  \cite{10.1111/rssb.12374}, which amounts to using the final weighted
  least squares IRLS iteration when fitting the models $\hat\eta_j(x_j)$.  A simple
  formula is available there, similar to what we got before, except a
  bit more detailed to accommodate observation weights. Again we can use
  Hadamard matrix operations to do this simultaneously for all $i$ and
  $j$.

  For the Cox model, we make a further approximation, since the implied
  weight matrix in IRLS is not diagonal. We simply use the diagonal
  which leads to an approximate Newton algorithm. Note that for the
  Cox model the intercept is always zero.

  \subsection{Software in R}
  \label{sec:software}
  We provide a \texttt{R} package \texttt{uniLasso} for fitting the
  models described in this paper. It mirrors the behavior of the 
  \texttt{glmnet} package for fitting lasso models. The function \texttt{cv.uniLasso}
  has arguments that include all those for \texttt{cv.glmnet} plus a
  few extras.
  This function does the following, currently for \texttt{"binomial"},
  \texttt{"gaussian"} and \texttt{"cox"} families.
  It does all the work outlined in Sections~\ref{sec:least-squares}--\ref{sec:gener-line-models}, and performs the second stage
  \texttt{cv.glmnet} using the (approximate) LOO $\hat\eta_j^{-i}$ as features.
  A \texttt{cv.glmnet} object contains information about suggested
  values of $\lambda$, as well as a fitted \texttt{glmnet} model fit
  to all the data with
  solutions at all values of $\lambda$. This object is referenced when
  one predicts from a \texttt{cv.glmnet} object. In this case we make
  sure that the coefficients that are stored on this objects for each $\lambda$ are the
  collapsed versions as in (\ref{eq:collapse2}).

  Hence \texttt{cv.uniLasso} returns a bona-fide \texttt{cv.glmnet}
  object, for which there are a number of methods for plotting,
  printing and making predictions.

  \texttt{cv.uniLasso} has three important arguments:
  \begin{enumerate}
  \item \texttt{loo = TRUE}: by default it uses the $\hat\eta_j^{-i}$
    as features. If this is set to \texttt{FALSE} it will use the
    univariate fitted values $\hat\eta_j^i$ instead.
  \item \texttt{lower.limits = 0}: this default choice guarantees
    that the $\hat\theta_j$ are all non-negative. We
    recommend this in order to get a sparser and more interpretable model.
  \item \texttt{standardize = FALSE}: this default is the most
    sensible to use, since part of the point here is to boost strong
    variables through their scale. Standardizing would undo that. 
  \end{enumerate}

  In addition we have \texttt{uniReg} and \texttt{cv.uniReg} for
  fitting the uniReg model, as well as \texttt{ci.uniReg} for
  computing bootstrap intervals for each coefficient. The function
  \texttt{polish.uniLasso} does the polishing as described in
  Section!\ref{sec:polish}. We also provide simulation functions for
  generating the data as used in this paper.

\section{Discussion}
\label{sec:discussion}
In this paper we have introduced a novel method for sparse regression
that ``stacks'' univariate regressions using a non-negative version of the lasso.
The procedure has interesting properties, namely that the final features weights have the same signs as the univariate coefficients and tend to be larger
when the  univariate coefficients are large.  The test set MSE of the
method is often similar to that of the lasso, with substantially smaller support and lower false positive rate.

There are many possible extensions for this work,  for  example, to
random forests, gradient boosting and neural networks. These are  topics for future study.

 %Another more adventurous application would be to neural networks: with features $x_1,x_2 \ldots x_p$,  the initial layer of a simple neural network gives weight $w_{jk}$ to $x_j$ as input into hidden unit $k$. We can restrict $w_{jk}$ to have the form $\theta_{jk}\hat\beta_{j}^{-i}$ with $\theta_{jk} \geq 0$, where as before $%\hat\beta_{j}^{-i}$ are the LOO univariate
%regression coefficients.

Current versions of the uniLasso package  can be found at 

\url{https://github.com/trevorhastie/uniLasso} and CRAN  (R)\\
\url{https://github.com/sophial05/uni-lasso}   (Python, written by
Sophia Lu).
 
\subsection*{Disclosure statement}
The authors have no conflicts of interest to declare.

\subsection*{Acknowledgments}
The authors thank two referees and the Editor for helpful comments that improved this manuscript.
We also thank Jerome Friedman, Alden Green, Sam
Gross, Chris Habron, Lucas Janson, Mert Pilanci, Jonathan Taylor, Ryan
Tibshirani and Hui Zou for helpful  discussions. We especially thank
Ryan for posing and proving theorem \ref{condthm2}, and Chris Habron for his contributions of Section 3.
% , and Alden for the suggestion to investigate adversarial
% attacks in Section \ref{sec:moreexamples}.
S.C. was supported by NSF grants DMS-2113242 and DMS-2153654.   T.H.
was supported by NIH grant R01 GM134483; R.T. was supported by NIH
grant  R01 GM134483 and NSF grant 19DMS1208164.

\appendix

\section{ Proof of Proposition 1}\label{sec:proof.-prop-1}
If you substitute $\eta_j^i =\hat\beta_{0j}+\hat\beta_jx_{ij} $ in the
objective in 
(\ref{eq:1}) you get 

$$\sum_{i=1}^n (y_i - (\theta_0+ \sum_{j=1}^p\hat\beta_{0j}) 
-\sum_{j=1}^px_{ij}\hat\beta_j\theta_j)^2 + \lambda\sum_{j=1}^p |\theta_j|$$

Now you change variables to $\gamma_0 = \theta_0+
\sum_{j=1}^p\hat\beta_{0j}$ and $\gamma_j=\hat\beta_j\theta_j$.

The result is immediate, since $$|\theta_j| = \frac{
  |\gamma_j|}{|\hat\beta_j|}$$
and
$$\{\mbox{sign($\gamma_j$)=sign($\hat\beta_j$) $\forall j$ }\} \equiv
\{\mbox{ $\theta_j \geq 0\;\forall j$} \}
$$
\qed

\section{Proof of theorems in Section~\ref{sec:theory}}\label{sec:proof-theor-sect}
We first prove Theorem \ref{mainthm}.
We will need the following result (see, e.g., \cite[Theorem A.7.1]{talagrand10}).
\begin{lmm}[Bernstein's inequality]
Let $Z_1,\ldots,Z_n$ be i.i.d.~random variables such that $\E(Z_1)=0$ and $\E(e^{\beta |Z_1|}) \le 2$ for some $\beta > 0$. Then for all $t\ge 0$,
\[
\P\biggl(\biggl|\frac{1}{n}\sum_{i=1}^n Z_i \biggr|\ge t\biggr) \le 2\exp\biggl(-\min\biggl\{\frac{n \beta^2 t^2}{4}, \frac{n \beta t}{2}\biggr\}\biggr).
\]
\end{lmm}
Define random variables
\begin{align*}
A &:= \frac{1}{n}\sum_{i=1}^n Y_i^2, \ \  B_j := \frac{1}{n}\sum_{i=1}^n X_{i,j}^2, \ \  D_j := \frac{1}{n}\sum_{i=1}^n Y_i X_{i,j}.
\end{align*}
As a corollary of Bernstein's inequality, we obtain the following.
\begin{cor}\label{boundcor}
There is a positive constant $C_3$ depending only on $C_1$ and $C_2$ such that for all $t\ge 0$, 
\[
\P(|A-\E(A)|\ge t) \ge 2e^{-C_3n \min\{t^2, t\}}, 
\]
ad the same bound holds for $\P(|B_j-\E(B_j)|\ge t)$ and $\P(|D_j-\E(D_j)|\ge t)$ for all $j$. 
\end{cor}
\begin{proof}
Let $Z_i := Y_i^2 -\E(A)$. Then $\E(Z_i)=0$, and for any $t\ge 0$, 
\[
\P(|Z_1| \ge t) \le \P(Y_1^2 \ge t) \le C_1 e^{-C_2t}. 
\]
Thus, for any $\beta \in (0,C_2)$, 
\begin{align*}
\E(e^{\beta |Z_1|}) &=\int_0^\infty \beta e^{\beta t} \P(|Z_1|\ge t)dt\\
&\le \int_0^\infty \beta e^{\beta t} C_1 e^{-C_2 t} dt\\
&= \frac{C_1 \beta }{C_2 - \beta}. 
\end{align*}
Choosing $\beta$ sufficiently small makes the upper bound $\le 2$, allowing us to apply Bernstein's inequality. A similar argument works for $B_j$. For $D_j$, take $Z_i := Y_i X_{i,j}-\E(D_j)$. Then note that $\E(Z_i)=0$, and for any $t\ge 2|\E(D_j)|$,
\begin{align*}
\P(|Z_1|\ge t) &\le \P(|Y_1X_{1,j}| \ge t/2) \\
&\le \P(Y_1^2 \ge t/2) + \P(X_{1,j}^2 \ge t/2)\\
&\le 2C_1 e^{-C_2 t/2}. 
\end{align*}
Since $|\E(D_j)|$ can be bounded above by a number that depends only on $C_1$ and $C_2$, this shows that there are constants $C_3$ and $C_4$ depending only on $C_1$ and $C_2$ such that for all $t\ge 0$,
\[
\P(|Z_1|\ge t) \le C_3 e^{-C_4 t}. 
\]
The proof is now completed by proceeding as before.
\end{proof}

We also obtain the following second corollary.
\begin{cor}\label{boundcor2}
There are positive constants $C_5$, $C_6$ and $C_7$ depending only on $C_0$, $C_1$ and $C_2$ such that for any $t\in [0, C_5]$ and any $i$ and $j$, $\P(| \hat\beta^{-i}_j-\beta_j |\ge t) $ and $\P(|\hat{\alpha}^{-i}_j-\alpha_j |\ge t) $ are bounded above by $ C_6 e^{-C_7n t^2}$. 
\end{cor}
\begin{proof}
Fix $i,j$. Let
\begin{align*}
&Q_1:= \frac{1}{n-1}\sum_{k\ne i}Y_k X_{k,j}, \ \ Q_2:= \frac{1}{n-1} \sum_{k\ne i}Y_k, \\
&Q_3:= \frac{1}{n-1} \sum_{k\ne i} X_{k,j}, \ \ Q_4 := \frac{1}{n-1}\sum_{k\ne i}X_{k,j}^2.%, \ \ Q_5:= \frac{1}{n-1}\sum_{k\ne i}X_{k,j}.
\end{align*}
Using the same approach as in Corollary \ref{boundcor} and the fact that $n-1 \ge n/2$ (because $n\ge 2$), we deduce the concentration inequality
\begin{align}\label{qlineq}
\P(|Q_l - \E(Q_l)| \ge t) \le 2e^{-Kn\min\{t^2,t\}}
\end{align}
for each $1\le l\le 4$. Now, note that $\E(Q_4) = \E(X_j^2)$ and $\E(Q_3) = \E(X_j)$. Moreover, recall that $\var(X_j) = \E(X_j^2)-(\E(X_j))^2 \ge C_0>0$. Combining these observations with the above inequality, we see that there are positive constants $K_1$ and $K_2$ such that 
\begin{align}\label{q4q3}
\P(|Q_4 -Q_3^2| < C_0/2) \le K_1 e^{-K_2n}.
\end{align}
Since 
\[
 \hat\beta^{-i}_j = \frac{Q_1 -Q_2Q_3}{Q_4 - Q_3^2}, 
\]
the inequality \eqref{q4q3} gives 
\begin{align}
\P(| \hat\beta^{-i}_j-\beta_j|\ge t) &= \P\biggl(\biggl|\frac{Q_1-Q_2Q_3 - (Q_4-Q_3^2)\beta_j}{Q_4-Q_3^2}\biggr|\ge t\biggr)\notag \\
&\le \P(|Q_1-Q_2Q_3 - (Q_4-Q_3^2)\beta_j|\ge C_0 t/2) +  K_1 e^{-K_2n}.\label{thetaineq}
\end{align}
Now take any $s>0$, and suppose that $|Q_l-\E(Q_l)|< s$ for $1\le l\le 4$. Then we also have 
\[
|Q_3^2 - (\E(Q_3))^2| \le |Q_3- \E(Q_3)|^2 + 2|\E(Q_3)||Q_3-\E(Q_3)| < s^2 + K_3s,
\]
where $K_3 $ depends only on $C_1$ and $C_2$. Similarly,
\[
|Q_2Q_3 - \E(Q_2)\E(Q_3)| < K_4s,
\]
where $K_4$ depends only on $C_1$ and $C_2$. Thus, we have
\begin{align*}
&|(Q_1-Q_2Q_3 - (Q_4-Q_3^2)\beta_j) - (\E(Q_1)-\E(Q_2)\E(Q_3) - (\E(Q_4)-(\E(Q_3))^2)\beta_j)|\\
&< K_5s+K_6s^2,
\end{align*}
where $K_5$ and $K_6$ depend only on $C_0$, $C_1$ and $C_2$. But
\begin{align*}
&\E(Q_1)-\E(Q_2)\E(Q_3) - (\E(Q_4)-(\E(Q_3))^2)\beta_j\\
&= \E(YX_j) - \E(Y)\E(X_j) - (\E(X_j^2) - (\E(X_j))^2) \frac{\E(YX_j)-\E(Y)\E(X_j)}{\E(X_j) -(\E(X_j))^2} = 0.
\end{align*}
Plugging this into the previous display shows that
\begin{align}\label{qineq}
\P(|Q_1-Q_2Q_3 - (Q_4-Q_3^2)\beta_j|\ge K_5s + K_6s^2)&\le \sum_{l=1}^4 \P(|Q_l- \E(Q_l)|\ge s)
\end{align}
Choosing $s$ to solve $K_5s + K_6s^2 = C_0t/2$, and combining equations \eqref{qlineq}, \eqref{thetaineq} and \eqref{qineq} yields the proof of the claimed concentration inequality for $ \hat\beta^{-i}_j$. The proof for $\hat{\alpha}^{-i}_j$ is similar. We omit the details.
\end{proof}
We are now ready to prove Theorem \ref{mainthm}.

\begin{proof}[Proof of Theorem \ref{mainthm}]
Throughout this proof, $K_1,K_2,\ldots$ will denote positive constants that may depend only on $C_0,C_1,C_2,\theta$ and $|S|$ and nothing else, whose values may change from line to line. Also, we will denote by $o(1)$ any random variable $Z$ such that 
\[
\P(|Z|\ge t) \le K_1 ne^{-K_2n t^2}
\]
for all $t\in [0,K_3]$ (for some $K_1,K_2,K_3$ according to the above convention), and by $O(1)$ any random variable $Z$ such that $|Z-K_1|=o(1)$ for some $K_1$. 
%For $1\le j\le p$, let
%\[
%\hat{\alpha}_j := \frac{1}{n}\sum_{i=1}^n \hat{\alpha}^{-i}_j. 
%\]
Define 
\begin{align*}
\tL(\theta_0,\theta_1,\ldots,\theta_p) :=  \frac{1}{n}\sum_{i=1}^n (Y_i - \theta_0 - \theta_1(\hat{Y}_{i,1}-\hat{\alpha}_1)- \cdots- \theta_p(\hat{Y}_{i,p}-\hat{\alpha}_p))^2 +\lambda \sum_{j=1}^p \theta_j. 
\end{align*}
Note that $(\hat{\theta}_0',\hat{\theta}_1,\ldots,\hat{\theta}_p)$ minimizes $\tL$ under $\hat{\theta}_j \ge0$ for all $1\le j\le p$ if and only if $(\hat{\theta}_0,\hat{\theta}_1,\ldots,\hat{\theta}_p)$ minimizes $L$ under the same constraint, where
\[
\hat{\theta}_0 = \hat{\theta}_0' - \sum_{j=1}^p \hat{\theta}_j \hat{\alpha}_j. 
\]
Note that this means $\hat{\gamma}_0=\hat{\theta}_0'$. 
We will henceforth consider this modified optimization problem. Note that for each $1\le j\le p$, and for any $\theta_0\in  \R$ and $\theta_1,\ldots,\theta_p\ge 0$, 
\begin{align}
\tL_j &:= \frac{\partial \tL}{\partial \theta_j } \notag \\
&= - \frac{2}{n}\sum_{i=1}^n (Y_i -\theta_0- \theta_1(\hat{Y}_{i,1}-\hat{\alpha}_1)- \cdots- \theta_p(\hat{Y}_{i,p}-\hat{\alpha}_p))(\hat{Y}_{i,j}-\hat{\alpha}_j) +\lambda.\label{ljform}
\end{align}
By an application of the Cauchy--Schwarz inequality, we get 
\begin{align*}
&\biggl|\frac{1}{n}\sum_{i=1}^n (Y_i - \theta_0 - \theta_1(\hat{Y}_{i,1}-\hat{\alpha}_1)- \cdots- \theta_p(\hat{Y}_{i,p}-\hat{\alpha}_p))(\hat{Y}_{i,j}-\hat{\alpha}_j)\biggr|\\
&\le \biggl(\frac{1}{n}\sum_{i=1}^n (Y_i -\theta_0- \theta_1(\hat{Y}_{i,1}-\hat{\alpha}_1)- \cdots- \theta_p(\hat{Y}_{i,p}-\hat{\alpha}_p))^2\biggr)^{1/2}\biggl(\frac{1}{n}\sum_{i=1}^n(\hat{Y}_{i,j}-\hat{\alpha}_j)^2\biggr)^{1/2}\\
&\le \sqrt{\tL(\theta_0,\ldots,\theta_p)} \biggl(\frac{1}{n}\sum_{i=1}^n(\hat{Y}_{i,j}-\hat{\alpha}_j)^2\biggr)^{1/2}.
\end{align*}
Now, note that 
\begin{align*}
\frac{1}{n}\sum_{i=1}^n(\hat{Y}_{i,j}-\hat{\alpha}_j)^2 &\le \frac{2}{n}\sum_{i=1}^n(\hat{Y}_{i,j}-\hat{\alpha}^{-i}_j)^2 +\frac{ 2}{n}\sum_{i=1}^n(\hat{\alpha}^{-i}_j-\hat{\alpha}_j)^2\\
&\le \frac{2}{n}\sum_{i=1}^n\hat{\beta}^2_{i,j}X_{i,j}^2 + \frac{4}{n}\sum_{i=1}^n (\hat{\alpha}^{-i}_j -\alpha_j)^2 + 4(\hat{\alpha}_j - \alpha_j)^2 \\
&\le \biggl(\frac{2}{n}\sum_{i=1}^n X_{i,j}^2\biggr) \max_{1\le i\le n}  (\hat\beta^{-i}_j)^2 + 4\max\{(\hat{\alpha}_j - \alpha_j)^2, \max_{1\le i\le n}(\hat{\alpha}^{-i}_j-\alpha_j)^2\}. 
\end{align*}
Let $M^\beta_j$ and $M^\alpha_j$ denote the two maxima on the right. Combining the previous three displays, we get that for $\theta_0\in \R$ and $\theta_1,\ldots,\theta_p\ge 0$, 
\begin{align}\label{ljineq}
\tL_j(\theta_0,\ldots, \theta_p) &\ge \lambda - \sqrt{\tL(\theta_0,\ldots,\theta_p)(2B_jM^\beta_j + 4M^\alpha_j )}.
\end{align}
Take any $j\in \{1,\ldots,p\} \setminus S$. Suppose that $\hat{\gamma}_j\ne 0$. Then $\hat{\theta}_j >0$. This implies that 
\[
\tL_j(\hat{\gamma}_0, \hat{\theta}_1,\ldots,\hat{\theta}_p) =0,
\]
because otherwise, we can slightly perturb $\hat{\theta}_j$ while maintaining the non-negativity constraint and decreasing the value of $L$. By inequality \eqref{ljineq}, this implies that if $\hat{\theta}_j >0$, then we must have
\[
\tL(\hat{\gamma}_0,\hat{\theta}_1,\ldots,\hat{\theta}_p)(2B_j M^\beta_j + 4M^\alpha_j ) \ge \lambda^2.
\]
But note that
\begin{align*}
\tL(\hat{\gamma}_0,\hat{\theta}_1,\ldots,\hat{\theta}_p) &\le \tL(0,0,\ldots,0) = A. 
\end{align*}
Thus, we conclude that if $\hat{\theta}_j >0$, then 
\[
\lambda^2 \le A(2B_j M^\beta_j + 4M^\alpha_j).
\]
Thus, 
\begin{align*}
\P(\hat{\theta}_j >0)&\le \P(A > 2\E(A)) + \P(B_j> 2\E(B_j)) \\
&\qquad + \P\biggl(M_j^\beta \ge \frac{\lambda^2}{16\E(A)\E(B_j)}\biggr) + \P\biggl(M_j^\alpha \ge \frac{\lambda^2}{16\E(A)}\biggr).
\end{align*}
By Corollary \ref{boundcor}, the first two probabilities on the right are bounded above by $K_1 e^{-K_2 n}$. Next, note that if 
\begin{align}\label{lambdacond}
K_3 |\beta_j|\le \lambda \le K_4,
\end{align}
for some sufficiently small $K_3$ and $K_4$, then by Corollary \ref{boundcor2},
\begin{align*}
\P\biggl(M_j^\beta \ge \frac{\lambda^2}{16\E(A)\E(B_j)}\biggr) &\le n  \P\biggl(\hat{\beta}_{1,j}^2\ge \frac{\lambda^2}{16\E(A)\E(B_j)}\biggr)\\
&\le n \P(|\hat{\beta}_{1,j} - \beta_j| \ge K_5 \lambda) \le K_6n e^{-K_7n\lambda^2}.
\end{align*}
Similarly, the same bound holds for the tail probability of $M_j^\alpha$. Combining, we get that under the condition \eqref{lambdacond},
\[
\P(\hat{\theta}_j > 0) \le K_8 ne^{-K_9n\lambda^2},
\]
and thus, 
\begin{align}\label{pebound}
\P(\hat{\gamma}_j\ne 0 \text{ for some } j\notin S\cup\{0\})&\le \P(E^c) \le K_8pn e^{-K_9n\lambda^2},
\end{align}
where $E$ denotes the event that $\hat{\theta}_j = 0$ for all $j\notin S\cup \{0\}$. Suppose that $E$ happens. Take any $k\in S$. If $\hat{\theta}_k \ne 0$, then $\tL_k(\hat{\gamma}_0,\hat{\theta}_1,\ldots,\hat{\theta}_p)=0$, whereas if  $\hat{\theta}_k = 0$, then $\tL_k(\hat{\gamma}_0,\hat{\theta}_1,\ldots,\hat{\theta}_p)\ge 0$ and $(\gamma_k - \hat{\gamma}_k)/\beta_k = \gamma_k/\beta_k > 0$. Thus, in either case, we have
\[
\frac{\gamma_k - \hat{\gamma}_k}{\beta_k}\tL_k(\hat{\gamma}_0,\hat{\theta}_1,\ldots,\hat{\theta}_p) \ge 0.
\]
By the formula \eqref{ljform} for $\tL_j$, this shows that 
\begin{align}
\frac{\lambda (\gamma_k - \hat{\gamma}_k)}{2\beta_k}&\ge \biggl(\frac{\gamma_k - \hat{\gamma}_k}{\beta_k}\biggr)\frac{1}{n}\sum_{i=1}^n\biggl(Y_i - \hat{\gamma}_0 - \sum_{j\in S} \hat{\theta}_j (\hat{Y}_{i,j} - \hat{\alpha}_j)\biggr)(\hat{Y}_{i,k} - \hat{\alpha}_k)\notag \\
&= \frac{\gamma_k - \hat{\gamma}_k}{n\beta_k}\sum_{i=1}^n \biggl[\gamma_0 - \hat{\gamma}_0 + \sum_{j\in S}(\gamma_j X_{i,j}   -  \hat{\theta}_j (\hat{Y}_{i,j} - \hat{\alpha}_j))\biggr] (\hat{Y}_{i,k} - \hat{\alpha}_k) \notag \\
&\qquad \qquad + \frac{\gamma_k - \hat{\gamma}_k}{n\beta_k}\sum_{i=1}^n \epsilon_i (\hat{Y}_{i,k} - \hat{\alpha}_k). \label{mainid}
\end{align}
%For each $j$, let
%\[
%\hat{\beta}_j := \frac{1}{n}\sum_{i=1}^n  \hat\beta^{-i}_j,
%\]
Define  
\[
M := \max_{1\le i\le n, \, j\in S} (| \hat\beta^{-i}_j - \beta_j| + |\hat{\beta}_j - \beta_j|). %\sqrt{M_j^\beta}.
\]
Then 
\begin{align}
|(\hat{Y}_{i,k} -\hat{\alpha}_k) - \beta_k X_{i,k} | &= |\hat{\beta}_{i,k} - \beta_k ||X_{i,k}|\le M|X_{i,k}|,\label{mainid2}
\end{align}
and since $\hat{\gamma}_j = \hat{\theta}_j \hat{\beta}_j$, 
\begin{align}
|\hat{\theta}_j (\hat{Y}_{i,j} -\hat{\alpha}_j) -\hat{ \gamma}_j X_{i,j} | &= |\hat{\theta}_j  \hat\beta^{-i}_jX_{i,j} - \hat{\theta}_j\hat{\beta}_j X_{i,j} | \notag \\
&= | \hat\beta^{-i}_j - \hat{\beta}_j||\hat{\theta}_jX_{i,j}|\le M|\hat{\theta}_j| |X_{i,j}|.\label{mainid3}
\end{align}
By \eqref{mainid2} and \eqref{mainid3}, we have
\begin{align}\label{mainid4}
&|\hat{\theta}_j (\hat{Y}_{i,j} - \hat{\alpha}_j)(\hat{Y}_{i,k} - \hat{\alpha}_k) - \hat{\gamma}_j \beta_k X_{i,j}X_{i,k}| \notag \\
&\le |\hat{\theta}_j (\hat{Y}_{i,j} -\hat{\alpha}_j) -\hat{ \gamma}_j X_{i,j} | |\hat{Y}_{i,k} - \hat{\alpha}_k| + |\hat{ \gamma}_j X_{i,j} ||(\hat{Y}_{i,k} -\hat{\alpha}_k) - \beta_k X_{i,k} |\notag \\
&\le M|\hat{\theta}_j| |X_{i,j}|(|\beta_k X_{i,k}| + M|X_{i,k}|) + M |\hat{\gamma}_j X_{i,k}X_{i,j}|
\end{align}
Let $S' := S\cup \{0\}$. Define $X_0\equiv 1$ and $X_{i,0} \equiv 1$ for $1\le i\le n$. Combining equations \eqref{mainid}, \eqref{mainid2}, \eqref{mainid3} and \eqref{mainid4}, we get
\begin{align}
&\frac{\lambda(\gamma_k-\hat{\gamma}_k)}{2\beta_k} -\biggl[ (\gamma_k - \hat{\gamma}_k) \sum_{j\in S'} (\gamma_j-\hat{\gamma}_j) \biggl(\frac{1}{n}\sum_{i=1}^n X_{i,j} X_{i,k}\biggr) + (\gamma_k - \hat{\gamma}_k)\frac{1}{n}\sum_{i=1}^n \epsilon_i X_{i,k}\biggr]\notag\\
&\ge - K_5 (M+M^2)\biggl|\frac{\gamma_k - \hat{\gamma}_k}{n\beta_k}\biggr|\sum_{i=1}^n \biggl(|\gamma_0-\hat{\gamma}_0||X_{i,k}|\notag \\
&\hskip2in +\sum_{j\in S}(1+ |\hat{\theta}_j| + |\hat{\gamma}_j|) (|X_{i,j}|+|\epsilon_i|)| X_{i,k}|\biggr). \label{mainid5}
\end{align}
Next note that $\tL_0(\hat{\gamma}_0,\hat{\theta}_1,\ldots,\hat{\theta}_p)=0$, where $\tL_0$ denotes the derivative of $\tL$ is the zeroth coordinate. This means that 
\begin{align*}
0 &= \frac{1}{n}\sum_{i=1}^n \biggl(Y_i -\hat{\gamma}_0-\sum_{j\in S} \hat{\theta}_j(\hat{Y}_{i,j}-\hat{\alpha}_j)\biggr)\\
&= \frac{1}{n}\sum_{i=1}^n \biggl[\gamma_0 - \hat{\gamma}_0 + \sum_{j\in S}(\gamma_j X_{i,j}   -  \hat{\theta}_j (\hat{Y}_{i,j} - \hat{\alpha}_j))\biggr] + \frac{1}{n}\sum_{i=1}^n \epsilon_i.
\end{align*}
Proceeding as before (and recalling that $X_{i,0}=1$), this gives
\begin{align}
&(\gamma_0 - \hat{\gamma}_0) \sum_{j\in S'} (\gamma_j-\hat{\gamma}_j) \biggl(\frac{1}{n}\sum_{i=1}^n X_{i,j} X_{i,0}\biggr) + (\gamma_0 - \hat{\gamma}_0)\frac{1}{n}\sum_{i=1}^n \epsilon_i X_{i,0} \notag \\
&\le \frac{M |\gamma_0-\hat{\gamma}_0|}{n}\sum_{j\in S} |\hat{\theta}_jX_{i,j}|. \label{mainid6}
\end{align} 
Let $\Delta := \max_{j\in S'}|\hat{\gamma}_j -\gamma_j|$. For each $j\in S$, let $\theta_j := \gamma_j/\beta_j$. Recall the quantity $M_2$ from the statement of the theorem. Note that if 
\begin{align}\label{mcond}
M < \frac{1}{2}M_2,
\end{align}
we get  that for each $j\in S$, 
\begin{align*}
|\hat{\theta}_j -\theta _j| &= \biggl|\frac{\hat{\gamma}_j}{\hat{\beta}_j}-\frac{\gamma_j}{\beta_j}\biggr|\\
&\le \frac{|\hat{\gamma}_j-\gamma_j|}{|\hat{\beta}_j|} + \gamma_j \biggl|\frac{1}{\hat{\beta}_j} - \frac{1}{\beta_j}\biggr|\\
&\le K_6 (\Delta + M).
\end{align*}
Combining this observation with the inequalities \eqref{mainid5} and \eqref{mainid6}, and summing over $k\in S'$, we get that 
\begin{align}\label{deltaineq}
\sum_{j,k\in S'} \hat{\sigma}_{j,k} (\gamma_j-\hat{\gamma}_j)(\gamma_k - \hat{\gamma}_k)  &\le K_7 \lambda \Delta + K_8 \Delta Q_1 + K_9 (M+M^2) \Delta(\Delta + M)Q_2, 
\end{align}
where 
\[
\hat{\sigma}_{j,k} := \frac{1}{n}\sum_{i=1}^n X_{i,j} X_{i,k}, \ \ Q_1 := \max_{j\in S'} \biggl|\frac{1}{n}\sum_{i=1}^n \epsilon_i X_{i,j}\biggr|, \ \  Q_2 := \max_{j\in S'} \frac{1}{n}\sum_{i=1}^n X_{i,j}^2.
\]
Let $\hat{\eta}$ be the smallest eigenvalue of the positive semidefinite matrix $(\hat{\sigma}_{j,k})_{j,k\in S'}$, so that
\[
\sum_{j,k\in S'} \hat{\sigma}_{j,k} (\gamma_j-\hat{\gamma}_j)(\gamma_k - \hat{\gamma}_k) \ge \hat{\eta} \sum_{j\in S'} (\hat{\gamma}_j -\gamma_j)^2 \ge \hat{\eta} \Delta^2.
\]
Combining this with equation \eqref{deltaineq} and rearranging, we get
\begin{align*}
(\hat{\eta} - K_9(M+M^2)Q_2)\Delta^2\le (K_7 \lambda + K_8  Q_1+ K_9 (M+M^2)MQ_2)\Delta 
\end{align*}
If the coefficient of $\Delta^2$ on the left is positive, this gives
\[
\Delta \le \frac{K_7 \lambda + K_8  Q_1+ K_9 (M+M^2)MQ_2}{\hat{\eta} - K_9(M+M^2)Q_2}. 
\]
Now, using Bernstein's inequality, it is easy to show that $Q_1 = o(1)$ and $Q_2 = O(1)$ (following the conventions introduced at the beginning of the proof).  By Corollary \ref{boundcor2}, we know that $M = o(1)$. Again by Bernstein's inequality, $\hat{\sigma}_{j,k} = \E(X_jX_k)+o(1)$ for each $j,k\in S$. From this, it is easy to deduce via standard matrix inequalities that $\hat{\eta}=\eta_0 + o(1)$, where $\eta_0$ is the minimum eigenvalue of $(\E(X_jX_k))_{j,k\in S'}$. Recalling that $X_0=1$, we see that this is equal to the minimum eigenvalue $\eta$ of the covariance matrix of $(X_j)_{j\in S}$. Combining all of these, it is now straightforward that under the condition~\eqref{lambdacond}, 
\[
\P(E \cap \{\Delta > K_{10} \lambda\})\le K_{11} n e^{-K_{12}n\lambda^2}. 
\]
Together with equation \eqref{pebound}, this completes the proof.
\end{proof}

\begin{proof}[Proof of Theorem \ref{condthm}]
Note that since $Y= \gamma_0 + \sum_{j\in S} \gamma_j X_j + \epsilon$, and $\epsilon$ is independent of the $X_j$'s, we have 
\begin{align*}
\beta_j &= \frac{\cov(Y, X_j)}{\var(X_j)} \\
&= \frac{\sum_{k\in S} \gamma_k\cov(X_k, X_j)}{\var(X_j)}\\
&= \gamma_j  + \sum_{k\in S\setminus\{j\}}\gamma_k \frac{\cov(X_k,X_j)}{\var(X_j)} = \gamma_j +\sum_{k\in S\setminus\{j\}} \gamma_k \delta_{k,j}. 
\end{align*}
Since $\gamma_j\ne 0$, we may divide throughout by $\gamma_j$ and get
\[
\frac{\beta_j}{\gamma_j} = 1 +\sum_{k\in S\setminus\{j\}} \frac{\gamma_k}{\gamma_j} \delta_{k,j}. 
\]
Now, by assumption, $\delta_{k,j}\ge 0$ whenever $\gamma_k/\gamma_j >0$, and $\delta_{k,j}\le 0$ whenever $\gamma_k /\gamma_j <0$. Thus,
\[
\sum_{k\in S\setminus\{j\}} \frac{\gamma_k}{\gamma_j} \delta_{k,j} \ge 0.
\]
This shows that $\beta_j/\gamma_j \ge 1$. In particular, $\beta_j$ is nonzero and has the same sign as $\gamma_j$.
\end{proof}

\begin{proof}[Proof of  Theorem~\ref{condthm2}]
Note that 
\begin{align*}
\beta_j &= \frac{\cov(Y, X_j)}{\var(X_j)}\\
&= \sum_{k\in S} \gamma_k \frac{\cov(X_k,X_j)}{\var(X_j)}\\
&= \sum_{k\in S}\gamma_k \delta_{kj}.
\end{align*}
Thus,
\begin{align*}
\frac{\beta_j}{\gamma_j} &= \sum_{k\in S} \frac{\gamma_k}{\gamma_j} \delta_{kj}\\
&= \sum_{k\in A_j} \biggl|\frac{\gamma_k}{\gamma_j} \delta_{kj}\biggr| - \sum_{k\notin A_j } \biggl|\frac{\gamma_k}{\gamma_j} \delta_{kj}\biggr|\\
&= \frac{\sum_{k\in A_j}|\gamma_k \delta_{kj}| - \sum_{k\notin A_j} |\gamma_k \delta_{kj}|}{|\gamma_j|}. 
\end{align*}
Since $\beta_j \gamma_j \ge 0$ if and only if $\beta_j/\gamma_j\ge 0$, this proves the claim.
\end{proof}

%%%%%%%%%%%%%

\section{Proof of  Theorem \ref{uniregthm}}\label{sec:proof-theor-refun}
Let $\gamma := (\beta_0, \beta)\in \R^{p+1}$ and define $\gols$ and $\gur$ similarly. Let $X$ denote the $n\times (p+1)$ matrix whose  columns are indexed from $0$ to $p$; $X_{i,0}=1$ for each $i$, and for $1\le j\le p$, $X_{i,j}$ is as before. Let $Y\in \R^n$ be the vector whose $i^{\mathrm{th}}$ component is $Y_i$, and define $\epsilon$ similarly. Note that $Y=X\gamma+\epsilon$, and that $\muols$ is the Euclidean projection on $Y$ on the column space $\mc$ of $X$. Let 
\[
\md:= \{X\theta:\theta = (\theta_0,\ldots,\theta_p)\in \R^p,\, \theta_j \buni_j\ge 0  \text{ for all } 1\le j\le p\},
\]
so that $\md$ is a random closed convex subset of $\mc$, and $\muur$ is the Euclidean projection of $Y$ on $\md$. Let $E$ be the event that $\buni_j \bbuni_j\ge 0$ for each $1\le j\le p$. Note that if $E$ happens, then $\mu\in \md$.

\begin{lmm}\label{elmm}
There is a positive constant $C_0$ depending only on $\lambda_0$, $\delta$ and $\sigma$ such that 
\[
\P(E) \ge 1- pe^{-C_0n}.
\]
\end{lmm}
\begin{proof}
If $\bbuni_j =0$, then $\buni_j \bbuni_j\ge 0$ anyway. Thus,
\begin{align*}
\P(E^c) &= \P(\buni_j \bbuni_j < 0 \text{ for some } j \text{ such that } \bbuni_j \ne0)\\
&\le \sum_{j\, : \, \bbuni_j \ne 0} \P(\buni_j \bbuni_j < 0). 
\end{align*}
Take any $j$ such that $\bbuni_j \ne 0$. Then by standard results, conditional on $X$, 
\begin{align*}
\buni_j \sim N\biggl(\bbuni_j, \frac{\sigma^2}{\sum_{i=1}^n (X_{i,j} - \overline{X}_j)^2}\biggr),
\end{align*}
where $\overline{X}_j := \frac{1}{n}\sum_{i=1}^n X_{i,j}$. Recall that $|\bbuni_j|\ge \delta$ and $\var(X_{i,j})\ge \lambda_0$. From this, it is easy to deduce that there is a positive constant $C_0$ depending only on $\lambda_0$, $\delta$ and $\sigma$ such that $\P(E^c)\le p e^{-C_0n}$. This completes the proof.
\end{proof}

\begin{lmm}\label{keylmm}
If the event $E$ happens, then
\begin{align*}
\|\mu - \muur\|^2 &\le \|\mu - \muols\|^2 -  \|\muols-\muur\|^2.
\end{align*}
\end{lmm}
\begin{proof}
Since $\muols$ is the projection of $Y$ on $\mc$ and $\mc$ is a subspace, it follows that  $Y-\muols$ is orthogonal to $x-\muols$  for all $x\in \mc$. Since $\md \subseteq \mc$, this implies that for any $\nu\in \md$,
\begin{align*}
\|Y-\nu\|^2 &= \|Y-\muols\|^2 + \|\muols - \nu\|^2.
\end{align*}
This shows that the projection of $Y$ on $\md$ (that is, $\muur$) is the same as the projection of $\muols$ on $\md$. Since $\md$ is convex and $\mu\in \md$ (because $E$ has happened), $t\mu + (1-t)\muur\in \md$ for all $t\in [0,1]$. But then, by the preceding sentence, 
\[
\|\muols - (t\mu + (1-t)\muur)\|^2\ge\|\muols - \muur\|^2
\]
for all $t\in [0,1]$. In particular, the derivative of the left side with respect to $t$ should be nonnegative at $t=0$. This shows that 
\[
(\muols - \muur)^T (\muur - \mu)\ge 0.
\]
From this, we get
\begin{align*}
\|\muols - \mu\|^2 - \|\muur-\mu\|^2 &= \|\muols -\muur\|^2 + 2(\muols - \muur)^T(\muur - \mu)\\
&\ge \|\muols-\muur\|^2.
\end{align*}
This completes the proof.
\end{proof}
\begin{lmm}\label{wishlmm}
Let $Z$ be an $n\times p$ matrix of i.i.d.~$N(0,1)$ random variables. Let $1\in \R^n$ be the vector of all $1$'s, and let $V := Z - \frac{1}{n}11^TZ$. Let $\mu$ be the minimum eigenvalue of $\frac{1}{n}V^TV$ and $\nu$ be the maximum eigenvalue of $\frac{1}{n}Z^TZ$. There are positive constants $C_1,C_2, \mu_0,\nu_0$ depending only on $r_0$ such that 
\[
\P(\mu\ge \mu_0, \, \nu \le \nu_0) \ge 1-C_1 e^{-C_2n}.
\]
\end{lmm}
\begin{proof}
By basic facts from statistics, we know that $Z^TZ$ follows the Wishart distribution $W_p(n, I_p)$, and $V^TV \sim W_p(n-1, I_p)$. The claimed bounds now follow from known results in the literature, such as \cite[Proposition 2.4 and Theorem 3.3]{rudelsonvershynin10}.
\end{proof}
\begin{lmm}\label{lambdalmm}
Let $\hl_0$ and $\hl_1$ be the minimum and maximum eigenvalues of $\frac{1}{n} X^T X$. There are positive constants $C_3,C_4,\lambda_0', \lambda_1'$ depending only on $\lambda_0$, $\lambda_1$, and $r_0$ such that 
\[
\P(\lambda_0'\le \hl_0\le\hl_1\le \lambda_1') \ge 1-C_3 e^{-C_4n}.
\] 
\end{lmm}
\begin{proof}
Let us write $X = [1 \, \, \, \tx]$, where $\tx$ consists of columns $1,\ldots,p$ of $X$ and $1$ denotes the first column, which has all $1$'s. Then for any $y = (y_0, \tilde{y})\in \R^{p+1}$, where $y_0$ denotes the first coordinate of $y$, we have
\begin{align}\label{xyeq}
\|Xy\|^2 = \|y_0 1 + \tx \tilde{y}\|^2.
\end{align}
Let $\Sigma^{1/2}$ denote the positive definite square-root of $\Sigma$. Then we can write $\tx = Z \Sigma^{1/2}$, where $Z$ is an $n\times p$ matrix with i.i.d.~$N(0,1)$ random variables. Let $V:= Z-\frac{1}{n}11^TZ$. Let $\mu$ be the smallest eigenvalue of $\frac{1}{n}V^TV$ and $\nu$ be the largest eigenvalue of $\frac{1}{n}Z^TZ$. By the above identity and the inequality $\|u+v\|^2\le 2\|u\|^2+2\|v\|^2$, we get
\begin{align*}
\|Xy\|^2 &= \|y_0 1 + Z \Sigma^{1/2} \tilde{y}\|^2 \\
&\le 2\|y_0 1\|^2 + 2\|Z\Sigma^{1/2}\tilde{y}\|^2\\
&\le 2y_0^2 n + 2n \nu \|\Sigma^{1/2}\tilde{y}\|^2\\
&\le 2y_0^2 n + 2n \nu \lambda_1 \|\tilde{y}\|^2 \le 2n(1+\lambda_1 \nu) \|y\|^2.
\end{align*}
This shows that 
\begin{align}\label{l1}
\hl_1 \le 2(1+\lambda_1 \nu).
\end{align}
Next, let $z$ be the projection of $Z \Sigma^{1/2} \tilde{y}$ on the span of $1$, given by
\[
z = n^{-1}11^T Z\Sigma^{1/2} \tilde{y}.
\]
Then by equation \eqref{xyeq},
\begin{align}
\|Xy\|^2 &= \|y_0 1- z\|^2 + \|\tx \tilde{y}-z\|^2\notag \\
&\ge \|\tx \tilde{y}-z\|^2 = \|Z \Sigma^{1/2} \tilde{y} - n^{-1}11^TZ\Sigma^{1/2} \tilde{y}\|^2\notag \\
&= \|V \Sigma^{1/2}\tilde{y}\|^2\ge n \mu \|\Sigma^{1/2}\tilde{y}\|^2 \ge n \mu \lambda_0 \|\tilde{y}\|^2.\label{xymain}
\end{align}
But by the first line of the above display, we also have
\begin{align}\label{xy1}
\|Xy\|^2 &\ge \|y_01 - z\|^2\ge (\|y_01\| - \|z\|)^2= (|y_0|\sqrt{n} - \|z\|)^2.
\end{align}
Now, note that 
\begin{align}\label{xy2}
\|z\| &= n^{-1} |1^T Z\Sigma^{1/2} \tilde{y}| \|1\| \le n^{-1/2} \|1\| \|Z\Sigma^{1/2}\tilde{y}\|\le \sqrt{n} \nu \lambda_1 \|\tilde{y}\|. 
\end{align}
So, if $ \nu \lambda_1 \|\tilde{y}\| \le \frac{1}{2}|y_0|$, then by \eqref{xy1} and \eqref{xy2}, 
\begin{align*}
\|Xy\|^2 &\ge \frac{1}{4}y_0^2 n = \biggl(\frac{1}{8} y_0^2 + \frac{1}{8}y_0^2\biggr) n \ge\biggl( \frac{1}{8} y_0^2 +\frac{4}{8}\nu^2 \lambda_1^2 \|\tilde{y}\|^2\biggr) n.
\end{align*}
Without loss, let us assume that $\nu$ and $\lambda_1$ were chosen so large that $4\nu^2 \lambda_1^2 \ge 1$. Then the above inequality gives
\begin{align}\label{xymin1}
\|Xy\|^2 \ge \biggl( \frac{1}{8} (\|y\|^2 - \|\tilde{y}\|^2)  +\frac{4}{8}\nu^2 \lambda_1^2 \|\tilde{y}\|^2\biggr) n\ge \frac{1}{8}\|y\|^2 n.
\end{align}
On the other hand, if $ \nu \lambda_1 \|\tilde{y}\| > \frac{1}{2}|y_0|$, then by equation \eqref{xymain},
\begin{align}
\|Xy\|^2 &\ge n\mu\lambda_0\biggl(\frac{1}{2}\|\tilde{y}\|^2 + \frac{1}{2}\|\tilde{y}\|^2\biggr)\notag\\
&\ge n\mu\lambda_0\biggl(\frac{1}{2}\|\tilde{y}\|^2 + \frac{1}{8\nu^2\lambda_1^2}y_0^2\biggr)\notag \\
&= n\mu\lambda_0\biggl(\frac{1}{2}\|\tilde{y}\|^2 + \frac{1}{8\nu^2\lambda_1^2}(\|y\|^2 - \|\tilde{y}\|^2)\biggr)\ge \frac{n\mu\lambda_0}{8\nu^2\lambda_1^2} \|y\|^2,\label{xymin2}
\end{align}
where the last inequality follows from the assumption that $4\nu^2 \lambda_1^2 \ge 1$. Combining  the inequalities \eqref{xymin1} and \eqref{xymin2}, we see that 
\begin{align}\label{l2}
\hl_0 \ge \min\biggl\{\frac{1}{8}, \, \frac{\mu\lambda_0}{8\nu^2\lambda_1^2}\biggr\}.
\end{align}
Inequalities \eqref{l1} and \eqref{l2} and Lemma \ref{wishlmm} complete the proof.
\end{proof}

We are now ready to prove Theorem \ref{uniregthm}. First, note that 
\begin{align*}
\|\muols - \muur\|^2 &= \|X\gols - X \gur\|^2 \ge n\hl_0 \|\gols - \gur\|^2\\
&=n\hl_0 (\aols - \aur)^2 +n \hl_0 \sum_{j=1}^p (\bols_j - \bur_j)^2.
\end{align*}
By the definition of $\bur$, $\bur_j \buni_j\ge 0$ for each $j$; that is, $\bur_j$ and $\buni_j$ are on the same side of zero on the real line (including the possibility that one or both may be equal to zero). Thus, if $E$ happens and $\bols_j \bbuni_j < 0$, then $\bols_j$ and $\bur_j$ are on opposite sides of zero (again, including the possibility that one or both may be equal to zero), and hence,
\[
(\bols_j - \bur_j)^2 \ge (\bols_j)^2.
\]
(Note that this does not hold if we weaken the condition to $\bols_j \bbuni_j \le 0$. This is because this weakening leaves open the possibility that $\bbuni_j = 0$ and $\bols_j,\bur_j$ are on the same side of zero.) Combining, we see that if $E$ happens, then
\begin{align}\label{mulower}
\|\muols - \muur\|^2  \ge n \hl_0 \sum_{j\, : \, \bols_j \bbuni_j < 0} (\bols_j)^2.
\end{align}
Take any $j$ such that $\beta_j=0$ and $\bbuni_j\ne 0$. Conditional on $X$,
\begin{align}\label{bolsj}
\bols_j \sim N(0,\sigma^2\theta_j),
\end{align}
where $\theta_0,\ldots,\theta_p$ are the diagonal elements of $(X^TX)^{-1}$. Since $\theta_j\ge (n\hl_1)^{-1}$, this shows that 
\begin{align}\label{mainlow}
\E((\bols_j)^2 1_{\{\bols_j \bbuni_j < 0\}}|X) = \frac{1}{2}\E((\bols_j)^2|X) \ge \frac{\sigma^2}{2n\hl_1}. 
\end{align}
Let $F$ be the event that $\lambda_0'\le \hl_0 \le \hl_1 \le \lambda_1'$, where $\lambda_0',\lambda_1'$ are the constants from Lemma~\ref{lambdalmm}, and let $H:=E\cap F$. Then by the inequality \eqref{mulower},
\begin{align}
\E\|\muols - \muur\|^2 &\ge \E(\|\muols - \muur\|^2 1_H) \notag \\
&\ge \E\biggl(n \hl_0 \sum_{j\, : \, \bols_j \bbuni_j < 0} (\bols_j)^2 1_H\biggr)\notag \\
&\ge n\lambda_0'\sum_{j\, :\,\beta_j =0, \, \bbuni_j\ne 0} \E( (\bols_j)^21_{\{\bols_j \bbuni_j < 0\}} 1_H).\label{main1}
\end{align}
Take any $j$ such that $\beta_j=0$ and $\bbuni_j\ne 0$. Note that 
\begin{align}
\E( (\bols_j)^21_{\{\bols_j \bbuni_j < 0\}} 1_H) &= \E( (\bols_j)^21_{\{\bols_j \bbuni_j < 0\}} 1_F) \notag \\
&\qquad - \E( (\bols_j)^21_{\{\bols_j \bbuni_j < 0\}} 1_{F\cap E^c}).\label{main2}
\end{align}
By Lemma~\ref{lambdalmm}, the inequality \eqref{mainlow}, and the fact that the event $F$ depends only on $X$, we get
\begin{align}
\E( (\bols_j)^21_{\{\bols_j \bbuni_j < 0\}} 1_F)  &= \E( \E((\bols_j)^21_{\{\bols_j \bbuni_j < 0\}} |X) 1_F) \notag \\
&\ge \frac{\sigma^2}{2n\lambda_1'} \P(F)\ge \frac{\sigma^2}{2n\lambda_1'} (1-C_3 e^{-C_4n}).\label{main3}
\end{align}
On the other hand, by the Cauchy--Schwarz inequality, equation \eqref{bolsj}, and Lemma \ref{elmm},
\begin{align}
\E( (\bols_j)^21_{\{\bols_j \bbuni_j < 0\}} 1_{F\cap E^c}) &\le [\E((\bols_j)^41_F) \P(E^c)]^{1/2}\notag \\
&\le [\E(\E((\bols_j)^4|X)1_F)]^{1/2} \sqrt{p}e^{-\frac{1}{2}C_0n}\notag \\
&\le \biggl(\frac{3\sigma^4}{n\lambda_0'}\biggr)^{1/2} \sqrt{p}e^{-\frac{1}{2}C_0n}\notag \\
&= \frac{\sqrt{3 r}\sigma^2}{\sqrt{\lambda_0'}}e^{-\frac{1}{2}C_0n}.\label{main4}
\end{align}
Combining equations \eqref{main1}, \eqref{main2}, \eqref{main3} and \eqref{main4}, we get
\begin{align*}
\E\|\muols - \muur\|^2 &\ge n\lambda_0' q \biggl( \frac{\sigma^2}{2n\lambda_1'} (1-C_3 e^{-C_4n}) - \frac{\sqrt{3 r}\sigma^2}{\sqrt{\lambda_0'}}e^{-\frac{1}{2}C_0n}\biggr)\\
&= \frac{\sigma^2q\lambda_0'}{2\lambda_1'} - C_3 \frac{\sigma^2q\lambda_0'}{2\lambda_1'} e^{-C_4n} - \sqrt{3r\lambda_0'} q\sigma^2 n e^{-\frac{1}{2}C_0n}. 
\end{align*}
By Lemma \ref{keylmm}, this completes the proof.

\section{Derivation supporting Section~\ref{sec:sparsity}}\label{sec:deriv-supp-sect}
We wish to establish when the correlations of the univariate LOO
fitted values with the response $y$ are likely to be positive. Thanks
to Chris Habron for outlining this analysis.

Assume we have standardized both the feature vectors $x_j$ and
response $y$ to have zero mean and unit variance. Since the analysis
focuses on individual features, we will drop the index $j$.

Formula (\ref{eqn:LOOres}) relates the LOO residuals to the OLS residuals:
\[
y_i - \hat{\eta}^{-i} = \frac{y_i - \hat{\eta}^i}{1 - H_{ii}}.
\]
Since both $x$ and $y$ are standardized, the vector of OLS fits is given by
\begin{eqnarray}
  \hat{\eta} &=& x\hat{\beta}\nonumber\\
  &=& x(x^\top y)/n.\label{eg:etahat}
\end{eqnarray}

With
$D=\mbox{diag}(\frac1{1-H_{11}},\frac1{1-H_{22}},\ldots,\frac1{1-H_{nn}})$
we can write the vector of LOO fits as
\begin{equation}
  \label{eq:etahatmi}
  \hat\eta^{loo} = (I-D)y+Dx(x^\top y)/n. 
\end{equation}

We look at the sample covariance between $y$ and  $\hat\eta^{loo}$
\begin{equation}
  \label{eq:scov}
\text{Cov}(y,\hat\eta^{loo}) = 1-(y^\top Dy)/n +(y^\top D
x)(x^\top y)/n^2.  
\end{equation}
We would like to know when this covariance is positive. (Note that it
is always non-negative if we replace $\hat\eta^{loo}$ with $\hat\eta$).
Since $x$ is standardized,
\[
H_{ii} = \frac{1}{n} + \frac{x_i^2}{n}
\]
and hence
\[
\frac{1}{1 - H_{ii}}= \frac{n}{n - 1 - x_i^2}
\]

To identify when this covariance becomes positive, we make the
approximation that $H_{ii}$ is constant, which implies $x_{i}^2=1$. This isn't exact, but for small correlations between $x$ and $y$, this introduces minimal error.
With this approximation
\[
\text{Cov}(y, \hat{\eta}^{loo}) \approx 1 - \frac{n}{n - 2} + \text{Cov}(y, x)^2 \cdot \frac{n}{n - 2}
\]

This becomes positive when
\[
\text{Cov}(y, x)^2 > \frac{2}{n}
\]

The critical value for a significance test for a correlation using the $t$-distribution is:
\[
r = \frac{t}{\sqrt{n - 2 + t^2}}
\]

We achieve this when $t = \sqrt{2}$, which corresponds to a two-sided $p$-value of 0.16 for large $n$.

\bibliographystyle{agsm}

%\bibliography{/Users/tibs/dropbox/texlib/tibs}
\bibliography{tibs}

\end{document}